\documentclass[11pt]{article}
\usepackage[letterpaper,margin=1in]{geometry}
\usepackage[T1]{fontenc}
\usepackage{lmodern,microtype}
\usepackage{amsmath,amssymb,amsthm,mathtools}
\usepackage{aliascnt,booktabs,enumitem}
\usepackage{xcolor,tikz}
\usetikzlibrary{arrows.meta,positioning,calc,fit,backgrounds}
\usepackage[colorlinks=true,linkcolor=blue!55!black,citecolor=blue!55!black,urlcolor=blue!55!black]{hyperref}
\newtheorem{theorem}{Theorem}[section]
\newaliascnt{lemma}{theorem}
\newtheorem{lemma}[lemma]{Lemma}\aliascntresetthe{lemma}
\newaliascnt{proposition}{theorem}
\aliascntresetthe{proposition}
\theoremstyle{definition}
\newaliascnt{definition}{theorem}
\aliascntresetthe{definition}
\newtheorem{ruleenv}{Reduction Rule}
\usepackage[nameinlink,noabbrev]{cleveref}
\crefname{lemma}{lemma}{lemmas}\Crefname{lemma}{Lemma}{Lemmas}
\crefname{proposition}{proposition}{propositions}\Crefname{proposition}{Proposition}{Propositions}
\crefname{definition}{definition}{definitions}\Crefname{definition}{Definition}{Definitions}
\crefname{ruleenv}{reduction rule}{reduction rules}\Crefname{ruleenv}{Reduction Rule}{Reduction Rules}
\newcommand{\PDFVS}{\textnormal{\textsc{Planar Directed Feedback Vertex Set}}}
\newcommand{\GSCA}{\textnormal{\textsc{Grouped Strong Connectivity Augmentation}}}
\newcommand{\problemstatement}[4]{%
  \par\medskip
  \begingroup
  \setlength{\fboxsep}{8pt}%
  \noindent\fbox{%
    \begin{minipage}{\dimexpr\linewidth-2\fboxsep-2\fboxrule\relax}
      \textbf{#1}\par\smallskip
      \textbf{Input.} #2\par\smallskip
      \textbf{Parameter.} #3\par\smallskip
      \textbf{Question.} #4
    \end{minipage}%
  }%
  \par
  \endgroup
  \medskip
}
\newcommand{\proofstep}[1]{\par\smallskip\noindent\textbf{#1}\enspace}
\newcommand{\reach}{\leadsto}
\DeclareMathOperator{\comp}{cc}
\setlist{itemsep=2pt,topsep=4pt}
\tikzset{v/.style={circle,draw,fill=white,inner sep=2pt,minimum size=5mm},
  keep/.style={v,fill=blue!12,draw=blue!65!black},
  arr/.style={-{Stealth[length=2mm]},thick},
  back/.style={arr,red!70!black,dashed},
  box/.style={draw,rounded corners,fill=blue!4,align=center,inner sep=7pt},
  every picture/.style={font=\small}}
\title{A Polynomial Kernel for Planar Directed Feedback Vertex Set}

\author{
Zimo Sheng$^{1}$, Mingyu Xiao$^{2}$\\[0.5em]
$^{1}$Anhui University\\
$^{2}$University of Electronic Science and Technology of China\\[0.5em]
\href{mailto:shengzimo2016@gmail.com}
{\texttt{shengzimo2016@gmail.com}}
\qquad
\href{mailto:myxiao@gmail.com}
{\texttt{myxiao@gmail.com}}\\[0.25em]
}

\date{}

\begin{document}
\hypersetup{pageanchor=false}
\begin{titlepage}
\maketitle
\thispagestyle{empty}

\begin{abstract}
\normalsize
The \textsc{Directed Feedback Vertex Set} problem (DFVS) asks
whether a digraph can be made acyclic by deleting at most $k$
vertices. Whether DFVS admits a polynomial kernel parameterized
by $k$ is a major open problem in kernelization, even for planar digraphs.
We resolve the planar case by giving a deterministic kernel
with $O(k^{66}\log^2 k)$ vertices and arcs.
Our algorithm proceeds in three stages. First, we apply
structural reduction rules to the input digraph, bounding
the number of directed faces and of vertices with more than two direction alternations among their incident arcs. Second, we pass to the planar dual,
where vertex deletion corresponds to adding groups of
reverse arcs to make each weakly connected component
strongly connected. The structural bounds in the first stage yield a small
retained vertex set in the dual. We then compress the
dual instance by identifying vertices with the same
distance records from this retained vertex set. The main
technical contribution is a directed-cut argument
showing that this identification preserves feasibility.
Finally, we transform the polynomial-size
dual instance back into an instance of \PDFVS{}
via a $3$-CNF encoding and a planar graph construction.
\end{abstract}
\end{titlepage}
\hypersetup{pageanchor=true}

\section{Introduction}
\label{sec:intro}

A directed feedback vertex set is a set of vertices whose removal
eliminates all directed cycles in a digraph. The corresponding
parameterized decision problem is defined as follows.

\problemstatement
  {\textnormal{\textsc{Directed Feedback Vertex Set}} (DFVS)}
  {A finite digraph $D$ and a nonnegative integer $k$.}
  {$k$.}
  {Does there exist a set $X\subseteq V(D)$ with $|X|\leq k$
   such that $D-X$ has no directed cycle?}

Restricting $D$ to planar digraphs gives \PDFVS{} (PDFVS).
DFVS is fixed-parameter tractable with respect to $k$, as
established by Chen, Liu, Lu, O'Sullivan, and Razgon~\cite{ChenEtAl2008}.
The question of polynomial kernelization asks for a different
kind of algorithm: can we reduce an instance, in polynomial
time, to an equivalent instance of the same problem whose size
is bounded by a polynomial in $k$~\cite{FominEtAl2019}?
Fixed-parameter tractability alone does not guarantee such a
polynomial bound.

Whether DFVS admits a polynomial kernel parameterized by $k$
is a longstanding open problem in
kernelization~\cite{LokshtanovMisraSaurabh2012,Dagstuhl12241,KratschWahlstrom2020,FPTSchool2014,Xiao2014,AgrawalEtAl2018,MnichVanLeeuwen2017,LokshtanovEtAl2025,BergougnouxEtAl2021,GrossmannEtAl2022,CrespelleEtAl2023,DirksEtAl2025,BonamyEtAl2018Treewidth}.
It is the first open problem listed in the monograph
\emph{Kernelization}~\cite{FominEtAl2019} and in Jansen's survey
\emph{Quo Vadis, Kernelization?}~\cite{Jansen2026QuoVadis}.
The question was also highlighted in Saurabh's talk at the 2017 school
\emph{Recent Advances in Parameterized Complexity}~\cite{Saurabh2017Future}
and in Marx's open-problem contribution to Dagstuhl Seminar~17041~\cite[Section~4.13, Open Problem~14]{Marx2017PolynomialKernels}.
Even the planar case has remained unresolved, as explicitly noted by
Bodlaender et al.~\cite{BodlaenderEtAlMeta2016}.

In this paper, we resolve the planar case by giving a
deterministic polynomial kernel. Both the size of the output
graph and its deletion budget are bounded by
$O(k^{66}\log^2 k)$.

\begin{theorem}\label{thm:main}
\PDFVS{} admits a deterministic kernel with
$O(k^{66}\log^2 k)$ vertices and arcs and output parameter
$k'=O(k^{66}\log^2 k)$.
\end{theorem}

Our construction combines structural reductions in the original
graph, compression of an equivalent augmentation problem in
the planar dual, and a return to planar vertex deletion.
The main technical step shows that vertices with the same
distance information can be merged in the dual instance
without changing the answer. We first place this result in
the context of earlier work and then describe the construction.

\subsection{Related Work}

\paragraph{Exact and parameterized algorithms.}
Razgon gave an exact algorithm for DFVS running in
$O^*(1.9977^n)$ time, where $n$ is the number of vertices,
improving on the $O^*(2^n)$ bound obtained by enumerating
vertex subsets~\cite{RazgonDFVS2007}.
Here $O^*$ suppresses polynomial factors.
Chen, Liu, Lu, O'Sullivan, and Razgon established fixed-parameter tractability
with a $4^k k!n^{O(1)}$-time algorithm~\cite{ChenEtAl2008},
resolving a longstanding open question. Xiong and Xiao later
simplified the algorithm and improved its running time to
$O(2^{o(k)}k!(n+m))$, where $m$ is the number of
arcs~\cite{XiongXiao2025}.
Very recently, Lokshtanov, Saurabh, and Xue
gave single-exponential 
algorithms for the problem in planar digraphs~\cite{LokshtanovSaurabhXue2026}.

For the undirected version, \textsc{Feedback Vertex Set} (FVS),
both exact and parameterized algorithms have undergone a
series of improvements. Deterministic exact algorithms
progressed from $O^*(1.8899^n)$~\cite{RazgonFVS2006} to
$O^*(1.7548^n)$~\cite{FominEtAlFVS2008} and
$O^*(1.7117^n)$~\cite{GaspersLee2017}. Randomized algorithms
achieved $O^*(1.6667^n)$ via monotone local
search~\cite{FominEtAlMLS2019} and subsequently
$O^*(1.6297^n)$~\cite{LiNederlof2022}.
For deterministic parameterized algorithms, the running-time
bounds improved from $O^*(5^k)$~\cite{ChenFVS2008} to
$O^*(3.83^k)$~\cite{CaoChenLiu2015},
$O^*(3.619^k)$~\cite{KociumakaPilipczuk2014}, and
$O^*(3.460^k)$~\cite{IwataKobayashi2021}.
Randomized bounds improved from $O^*(4^k)$~\cite{BeckerEtAl2000}
to $O^*(3^k)$ using Cut \& Count~\cite{CyganEtAl2022}, and
then to $O^*(2.7^k)$~\cite{LiNederlof2022}.

\paragraph{Kernels for undirected feedback vertex set.}
For undirected FVS, polynomial kernelization has developed
from an existence result into a sequence of increasingly
small kernels. Burrage, Estivill-Castro, Fellows, Langston, Mac, and Rosamond~\cite{BurrageEtAl2006} gave a
kernel with $O(k^{11})$ vertices, and Bodlaender and van
Dijk~\cite{BodlaenderVanDijk2010} reduced this bound to
$O(k^3)$. Thomass\'e~\cite{Thomasse2010} brought the bound
down to $4k^2$, followed by Iwata's kernel with
$2k^2+k$ vertices~\cite{Iwata2017}.

Planarity allows an even stronger reduction in size: undirected
planar FVS admits a linear kernel. Starting with the
$112k$-vertex kernel of Bodlaender and
Penninkx~\cite{BodlaenderPenninkx2008}, subsequent work reduced
the bound to fewer than $97k$ vertices~\cite{AbuKhzamBouKhuzam2012},
then $29k$~\cite{Xiao2014}, and finally $13k$ via region
decomposition~\cite{BonamyKowalik2016}.
This progress contrasts with the directed setting, where
even the existence of a polynomial kernel for planar graphs
has remained unresolved. Our result establishes this
existence, leaving substantial room to improve the size bound.

\paragraph{Planarity and partial results for DFVS.}
For the arc-deletion version on planar digraphs, a minimum
set of arcs whose deletion makes the graph acyclic can be
found in polynomial time by planar cycle--cut duality and the
Lucchesi--Younger theorem~\cite{LucchesiYounger1978,Erickson2020,BorradaileKlein2009}.
The corresponding vertex-deletion
problem remains NP-hard on planar
digraphs~\cite{BermanYaroslavtsev2012}.

Previous kernelization results for DFVS use other parameters
or impose additional restrictions on the input.
Bergougnoux, Eiben, Ganian, Ordyniak, and Ramanujan obtained polynomial kernels parameterized
by the size of an undirected feedback vertex set, including
a linear bound on graphs of fixed genus~\cite{BergougnouxEtAl2021}.
This parameter can be much larger than the directed deletion
parameter $k$: an acyclic digraph may have arbitrarily many
vertex-disjoint cycles in its underlying undirected graph.
Dirks, Gerhard, Grobler, Mouawad, and Siebertz gave polynomial kernels for digraphs with bounded
induced directed cycle length~\cite{DirksEtAl2025}, whereas
planar digraphs can contain arbitrarily long induced directed
cycles. These results therefore do not give a polynomial
kernel for planar DFVS parameterized by $k$.

\paragraph{Connections to directed cut problems.}
DFVS also has an algorithmic connection to directed separation
problems. The algorithm of Chen, Liu, Lu, O'Sullivan, and Razgon~\cite{ChenEtAl2008}
reduces a compression subproblem to \textsc{Skew Separator},
which asks for a small vertex set destroying source-to-sink
paths prescribed by an ordering of the terminals.
In the vertex-deletion formulation, \textsc{Directed Multiway
Cut} asks for a small set of nonterminal vertices whose deletion
destroys all directed paths between distinct terminals.
\textsc{Directed Multicut} instead specifies ordered terminal
pairs and requires destroying every path from the source to
the sink of each pair.

The kernelization landscape for these problems differs from
that of DFVS. Parameterized by the cutset size,
\textsc{Directed Multiway Cut} is fixed-parameter
tractable~\cite{ChitnisHajiaghayiMarx2013}, but admits no
polynomial kernel or polynomial compression even with two
terminals, unless $\mathrm{NP}\subseteq\mathrm{coNP}/\mathrm{poly}$;
this holds for both vertex and edge
deletion~\cite{CyganEtAl2014Incompressibility}.
\textsc{Directed Multicut} is $\mathrm{W}[1]$-hard even with
four terminal pairs~\cite{PilipczukWahlstrom2018}, ruling out
a kernel parameterized by the cutset size unless
$\mathrm{FPT}=\mathrm{W}[1]$. These are results for general
digraphs and do not by themselves settle the kernelization
complexity of the planar restrictions.

\subsection{Technical Overview}

Our approach is initially motivated by one of our previous results (Theorem~2 in~\cite{BaiCaoXiao2026}), which gives a polynomial-time
algorithm for DFVS on planar digraphs admitting a \emph{bipolar
embedding}: the incoming arcs at each vertex occur consecutively,
as do the outgoing arcs. 
The proof splits each vertex into
an incoming and an outgoing copy, preserving planarity, and
reduces vertex deletion to arc deletion.
So, we come up with an idea to solve the problem by bounding the number of vertices without the above property (called special vertices later). 

For a fixed plane embedding, let $\sigma(v)$ count the changes
between incoming and outgoing arcs in the cyclic order around
$v$. We call $v$ \emph{ordinary} if $\sigma(v)=2$ and
\emph{special} if $\sigma(v)>2$. After removing vertices and
arcs not in any directed cycles, these two cases cover all
remaining vertices. Theorem 2 in~\cite{BaiCaoXiao2026} therefore identifies the case
with no special vertices as polynomial-time solvable. This
observation suggests our strategy: first bound the number of
special vertices and their total alternation count, and then
compress the remaining ordinary part. %through the planar dual.

Our proof proceeds in three stages. We first simplify the
original digraph by structural reductions, which will bound the number of special vertices. We then pass to
the planar dual and compress an equivalent arc-augmentation
instance. This step will compress the ordinary vertices.
Finally, we transform the compressed instance back
into a planar vertex-deletion instance of the original DFVS. The middle stage is
described in two steps below: the dual transformation and
the compression that makes its output small.

In the dual, deleting a vertex is represented by adding a
group of reverse arcs. Selecting groups should make each
weakly connected component strongly connected, with the
total number of selected groups bounded by the deletion
budget. This leads to the following intermediate problem.

\problemstatement
  {\textnormal{\textsc{Grouped Strong Connectivity Augmentation}} (GSCA)}
  {A finite digraph $H$, a family $\mathcal G$ of arc groups,
   and a nonnegative integer $k$.
   Each group is a set of arcs absent from $H$, with all
   endpoints in the same weakly connected component of $H$.
   Different groups may share arcs.}
  {$k$.}
  {Can we select at most $k$ groups from $\mathcal G$ so that
   adding all their arcs to $H$ makes every weakly connected
   component strongly connected?}

Each selected group has unit cost, regardless of how many arcs
it contains. When measuring the representation size, a
\emph{group-arc entry} is one occurrence of an arc in a group.
An arc shared by two groups contributes two entries. Groups
are indexed, so different groups may contain the same arc set.

%Write $\mathcal I_0=(D,k)$ for the input instance, with $D=(V,A)$, and $\mathcal I_1=(D_1,k_1)$ for the instance after the structural reductions, where $k_1\leq k$.

Figure~\ref{fig:overview} summarizes the construction and main steps of our kernelization algorithm.
Reading the upper row from left to right, we simplify the
input instance $\mathcal I_0$ to obtain $\mathcal I_1$ by applying some structural reductions,
transform it into the dual augmentation instance
$\mathcal I_2$, and compress it to obtain $\mathcal I_3$.
The lower row, read from right to left, shows the return
to planar vertex deletion: we encode $\mathcal I_3$ as a
$3$-CNF formula $\mathcal I_4$, planarize it to obtain
$\mathcal I_5$, and construct the output PDFVS instance
$\mathcal I_6$. Each box records the size bounds established
at that stage. We explain these steps below.

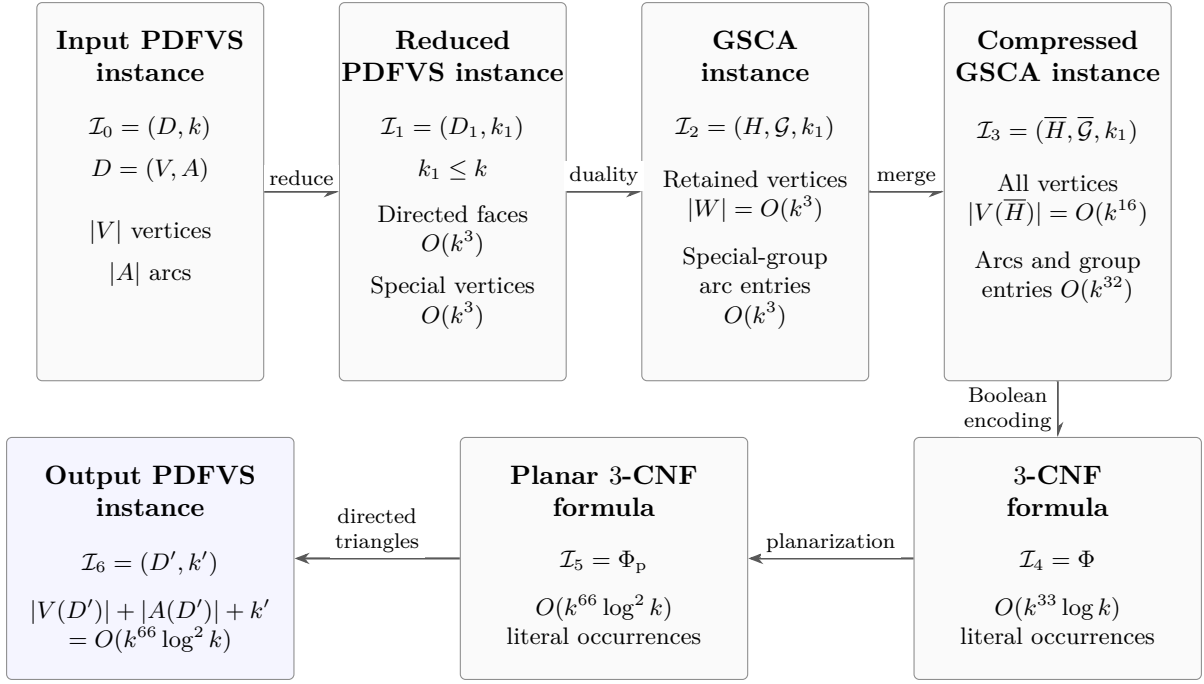
\begin{figure}[!t]
\centering
\begin{tikzpicture}[
  instance/.style={draw=black!45,rounded corners=2pt,fill=black!2,
    minimum width=3cm,minimum height=5cm},
  returninstance/.style={instance,minimum width=3.8cm,minimum height=3.2cm},
  heading/.style={anchor=north,align=center,font=\small\bfseries},
  detail/.style={anchor=north,align=center,font=\footnotesize},
  flow/.style={-{Stealth[length=2mm]},draw=black!65,line width=.8pt},
  action/.style={font=\scriptsize,align=center,fill=white,inner sep=2pt}
]
\node[instance] (input) at (0,0) {};
\node[instance] (reduced) at (4,0) {};
\node[instance] (dual) at (8,0) {};
\node[instance] (merged) at (12,0) {};
\node[returninstance] (cnf) at (12,-4.85) {};
\node[returninstance] (planar) at (6,-4.85) {};
\node[returninstance,fill=blue!4] (output) at (0,-4.85) {};

\node[heading] at ([yshift=-7pt]input.north)
  {Input PDFVS\\instance};
\node[heading] at ([yshift=-7pt]reduced.north)
  {Reduced\\PDFVS instance};
\node[heading] at ([yshift=-7pt]dual.north)
  {GSCA\\instance};
\node[heading] at ([yshift=-7pt]merged.north)
  {Compressed\\GSCA instance};
\node[heading] at ([yshift=-7pt]cnf.north)
  {$3$-CNF\\formula};
\node[heading] at ([yshift=-7pt]planar.north)
  {Planar $3$-CNF\\formula};
\node[heading] at ([yshift=-7pt]output.north)
  {Output PDFVS\\instance};

\node[detail] at ([yshift=-39pt]input.north) {
  \(\mathcal I_0=(D,k)\)\\[5pt]
  \(D=(V,A)\)\\[12pt]
  \(|V|\) vertices\\[5pt]
  \(|A|\) arcs
};
\node[detail] at ([yshift=-39pt]reduced.north) {
  \(\mathcal I_1=(D_1,k_1)\)\\[5pt]
  \(k_1\leq k\)\\[6pt]
  Directed faces\\\(O(k^3)\)\\[5pt]
  Special vertices \\ \(O(k^3)\)
};
\node[detail] at ([yshift=-39pt]dual.north) {
  \(\mathcal I_2=(H,\mathcal G,k_1)\)\\[9pt]
  Retained vertices\\\(|W|=O(k^3)\)\\[7pt]
  Special-group\\arc entries\\\(O(k^3)\)
};
\node[detail] at ([yshift=-39pt]merged.north) {
  \(\mathcal I_3=(\overline H,\overline{\mathcal G},k_1)\)\\[9pt]
  All vertices\\\(|V(\overline H)|=O(k^{16})\)\\[7pt]
  Arcs and group\\entries \(O(k^{32})\)
};
\node[detail] at ([yshift=-39pt]cnf.north) {
  \(\mathcal I_4=\Phi\)\\[7pt]
  \(O(k^{33}\log k)\)\\literal occurrences
};
\node[detail] at ([yshift=-39pt]planar.north) {
  \(\mathcal I_5=\Phi_{\mathrm p}\)\\[7pt]
  \(O(k^{66}\log^2 k)\)\\literal occurrences
};
\node[detail] at ([yshift=-39pt]output.north) {
  \(\mathcal I_6=(D',k')\)\\[7pt]
  \(|V(D')|+|A(D')|+k'\)\\
  \(=O(k^{66}\log^2 k)\)
};

\draw[flow] (input.east)--node[action,above]{reduce}(reduced.west);
\draw[flow] (reduced.east)--node[action,above]{duality}(dual.west);
\draw[flow] (dual.east)--node[action,above]{merge}(merged.west);
\draw[flow] (merged.south)--node[action,left]{Boolean\\encoding}(cnf.north);
\draw[flow] (cnf.west)--node[action,above]{planarization}(planar.east);
\draw[flow] (planar.west)--node[action,above]{directed\\triangles}(output.east);
\end{tikzpicture}
\caption{The sequence $\mathcal I_0,\ldots,\mathcal I_6$.
Read the upper row from left to right, then the lower row
from right to left. Each box gives the instance and the
quantities bounded at that stage. The last three arrows
convert the compressed GSCA instance back to PDFVS.
All bounds use the original input parameter $k$.}
\label{fig:overview}
\end{figure}

\paragraph{1. Reducing the original graph.}
The transition from $\mathcal I_0$ to $\mathcal I_1$ bounds
the structural quantities that will control the dual instance.
A face is \emph{directed} if its entire boundary can be
traversed following the arc directions. Let $F$ be the total
number of directed faces, counted separately in the embeddings
of the strongly connected components. Every solution must
delete a vertex on the boundary of each directed face.
If many such faces meet at a vertex, avoiding that vertex
requires other vertices to meet all those boundaries.
Our reduction rules bound how many of these faces any other
vertex can meet, yielding
\[
F\leq k_1^2(k_1+2)=O(k^3).
\]

Counting angles and applying Euler's formula relates the
alternation counts to the directed faces and gives
\[
\bigl|\{v:\sigma(v)>2\}\bigr|=O(k^3),
\qquad
\sum_{v:\sigma(v)>2}\sigma(v)=O(k^3).
\]
The face bound will control the sources and sinks in the dual.
The alternation bound will control the number of arcs needed
to represent deletions of special vertices.

\paragraph{2. Turning vertex deletion into arc addition.}
For each strongly connected component of $D_1$, take its
planar dual. Their disjoint union forms the base digraph $H$
of a GSCA instance $\mathcal I_2=(H,\mathcal G,k_1)$.
A vertex of $D_1$ corresponds to a face in its component's
dual. Instead of deleting that vertex, we add the reversals
of all arcs along the corresponding dual face boundary.
These reversals form one group. By planar cycle--cut
duality~\cite{Erickson2020,BorradaileKlein2009}, a set of vertex deletions makes
$D_1$ acyclic exactly when the corresponding groups make
every component of $H$ strongly connected. Moreover, strong
connectivity of each primal component makes its dual acyclic.

For an ordinary vertex, the corresponding dual face boundary
consists of two directed paths with the same start and end.
Adding a single reverse arc from the end to the start has
the same reachability effect as adding the original group:
all vertices on the two paths become mutually reachable.
We replace the group by this one arc and call it an
\emph{ordinary group}. For a special vertex $v$, the same
principle reduces its group to at most $\sigma(v)$ reverse
arcs. This is a \emph{special group}. The alternation bound
therefore limits the total number of special-group arc
entries to $O(k^3)$.

Let $W$ consist of all sources and sinks of $H$, together
with the endpoints of all reverse arcs in special groups.
These are the vertices we retain during compression.
Sources and sinks in each dual correspond to directed faces
of the original component. Together with the bound on
special-group arcs, this gives
\[
|W|=O(k^3),
\qquad
\sum_{\substack{B\in\mathcal G\\B\text{ special}}}|B|=O(k^3).
\]
At this point, the vertices outside $W$ and the ordinary
groups can still be arbitrarily numerous. The next step
bounds the size of the entire augmentation instance.

\paragraph{3. Merging vertices with the same distance record.}
We construct $\mathcal I_3=(\overline H,\overline{\mathcal G},k_1)$ by merging
vertices separately within each weakly connected component
$C$ of $H$. Put $W_C=W\cap V(C)$.
For $a\in W_C$ and $v\in V(C)$, let $d(a,v)$ be the minimum
number of ordinary reverse arcs used by a path from $a$ to
$v$, when base arcs and all ordinary reverse arcs are available.
Base arcs have cost zero, and $d(a,v)=\infty$ if no such
path exists. The distance record of $v$ is
\[
\bigl(\min\{d(a,v),k_1+1\}\bigr)_{a\in W_C}.
\]
Thus the record distinguishes distances up to $k_1$ and
uses the value $k_1+1$ for all larger distances, including
infinity. We keep each retained vertex separately and merge
vertices outside $W_C$ that have the same record.

Distance and neighborhood descriptions have also been useful
in domination kernels on sparse
graphs~\cite{DrangeEtAl2016,EickmeyerEtAl2017}.
Here we apply the directed-ball theorem of Le and
Wulff-Nilsen~\cite{LeWulffNilsen2024} to bound the number of
records in a component by
\[
O\bigl((|W_C|(k_1+1)+1)^4\bigr).
\]
Summing over the components and substituting $|W|=O(k^3)$ into the fourth-power bound yields
$O((k^3\cdot k)^4)=O(k^{16})$ vertices in total after merging.
After duplicates are removed, base arcs and ordinary groups
are specified by ordered pairs of these vertices and hence
contribute $O(k^{32})$ entries. Special groups retain their
total $O(k^3)$ arc entries. The resulting instance
$\mathcal I_3$ therefore has $O(k^{16})$ vertices and
$O(k^{32})$ arcs and group-arc entries in total.

The main difficulty is proving that merging preserves the
answer, since it may create new paths. Any solution before
merging remains a solution afterwards. For the converse,
consider one component in which a solution after merging
selects $t$ ordinary groups in addition to its special groups.
Add the same special groups before merging. We show that
at most $t$ additional available reverse arcs suffice to
make the original component strongly connected. Each such
arc can be supplied by one containing group, giving a
solution with no more groups than the selected solution
after merging.

Suppose, for a contradiction, that more than $t$
available reverse arcs were necessary. To use the standard
reverse-arc augmentation model, we first add auxiliary forward
arcs that preserve reachability, so that every available reverse arc
is the reverse of an existing arc. We can then apply the Lucchesi--Younger
theorem~\cite{LucchesiYounger1978} to obtain $t+1$
directed-cut requirements: each comes from a nonempty proper vertex
set with no entering arc, and each available reverse arc enters at
most one of these sets.
We represent the cuts by tree edges and assign vertices
to positions in the tree. The distance records allow us
to adjust these positions so that each class to be merged
lies entirely inside or outside each set. Retained vertices
keep both sides of every cut nonempty, while each reverse
arc still enters at most one set. The $t$ selected ordinary
groups, each consisting of one arc, cannot meet all $t+1$
requirements. This contradicts feasibility after merging.
Applying the argument to each component proves that
$\mathcal I_2$ and $\mathcal I_3$ are equivalent.

\paragraph{4. Returning to vertex deletion.}
We complete the construction through
$\mathcal I_3\to\mathcal I_4\to\mathcal I_5\to\mathcal I_6$.
First, we encode the choice of at most $k_1$ groups and the
resulting strong connectivity in each component as a
$3$-CNF formula $\mathcal I_4=\Phi$.
There are $O(k^{32})$ available groups, and the group-selection
and connectivity constraints can be encoded using
$O(k^{33}\log k)$ literal occurrences.

Lichtenstein's planarization~\cite{Lichtenstein1982} increases
this number quadratically, producing a planar formula
$\mathcal I_5=\Phi_{\mathrm p}$ with $O(k^{66}\log^2 k)$
literal occurrences. A construction using directed triangles
then transforms $\mathcal I_5$ into a PDFVS instance
$\mathcal I_6=(D',k')$ with only a linear increase in size.
Its vertex count, arc count, and output parameter are all
$O(k^{66}\log^2 k)$, as stated in \Cref{thm:main}.
All output vertices are unweighted and can be deleted.

Section~\ref{sec:prelim} introduces the notation and tools
used in the proof. Sections~\ref{sec:primal}, \ref{sec:dual},
and~\ref{sec:compression} give the structural reductions,
the dual transformation, and the compression argument,
respectively. Section~\ref{sec:return} gives the return
construction and completes the proof of \Cref{thm:main}.

\section{Preliminaries}
\label{sec:prelim}
We fix the graph notation and recall the facts about planar
embeddings and directed cuts used throughout the paper.
Tools specific to distance records and Boolean formulas are
introduced in Sections~\ref{sec:distance-rule} and~\ref{sec:return},
respectively, where they are used.

\subsection{Digraphs and Connectivity}
All logarithms are to base two. All graphs are finite.
For a digraph $D$, we write $V(D)$
and $A(D)$ for its vertex and arc sets. An arc $u\to v$
leaves $u$ and enters $v$. It is a loop if $u=v$.
Unless stated otherwise, we allow loops, repeated arcs with
the same direction, and opposite arcs $u\to v$ and $v\to u$.
An incoming neighbor of $v$ is a vertex $u$ with an arc
$u\to v$. Outgoing neighbors are defined symmetrically.
The outdegree $d^+(v)$ counts outgoing arcs with multiplicity,
so it need not equal the number of distinct outgoing neighbors.

For $X\subseteq V(D)$, the induced subgraph $D[X]$ contains
the vertices of $X$ and all arcs with both endpoints in $X$.
We write $D-X=D[V(D)\setminus X]$, and $D-v$ when $X=\{v\}$.
For a set $B$ of additional arcs with endpoints in $V(D)$,
$D+B$ denotes the digraph obtained by adding these arcs.
The reverse of $u\to v$ is $v\to u$; adding a reverse arc
retains the original arc.

A directed walk follows the arc directions and may repeat
vertices and arcs. A directed path is a directed walk with no
repeated vertex. We write $u\reach v$ if $v$ is reachable from $u$,
allowing a path of length zero. A directed cycle is a nonempty
closed directed walk with no repeated vertex other than its
first and last vertex. In particular, a loop and a pair of
opposite arcs are cycles of lengths one and two, respectively.
A digraph is acyclic if it contains no directed cycle.

The underlying undirected graph $G$ of $D$ is obtained by
forgetting arc directions and retaining one edge for every arc.
Distinct arcs therefore remain distinct edges. In particular,
opposite arcs give parallel edges even if $D$ has no repeated
arcs with the same direction. We write $\comp(G)$ for the
number of connected components of $G$, counting each isolated
vertex as one component.

A digraph is \emph{weakly connected} if its underlying
undirected graph is connected, and \emph{strongly connected}
if every vertex can reach every other vertex. Its weakly
connected components are the connected components of its
underlying graph. Its strongly connected components are the
maximal sets of mutually reachable vertices. When a component
is used as a set, we mean its vertex set. Every strongly
connected component lies in a unique weakly connected component,
but a weakly connected component may contain several strongly
connected components. Deleting vertices commutes with forgetting
arc directions, so $\comp(G-X)$ counts the weakly connected
components of $D-X$.

The \emph{condensation} of a digraph replaces each strongly
connected component by one vertex and retains the arcs between
distinct components. It is acyclic. A source has no incoming
arc, and a sink has no outgoing arc. In a finite acyclic digraph,
every vertex is reachable from a source and can reach a sink.
Every directed cycle lies in a single strongly connected component.

A tree is a connected undirected graph without a cycle.
It has a unique path between any two vertices, and deleting
an edge separates it into two components. A \emph{directed tree}
is obtained by orienting each edge of a tree. Its edges need
not all point towards or away from a common root.

\subsection{Plane Embeddings and Duality}
A graph is \emph{planar} if it admits a crossing-free drawing
in the plane, and a \emph{plane graph} is a graph with a fixed
such embedding. A digraph is planar when its underlying
undirected multigraph is planar. We compute an embedding in
polynomial time~\cite{HopcroftTarjan1974} and keep it fixed
when discussing faces and the cyclic order of arcs around vertices.

The faces are the connected regions of the complement of the
drawing, including the unbounded face. In a connected plane
graph with at least one edge, each face boundary is a closed
walk and may repeat vertices or edges. Incidences along this
boundary are counted with multiplicity. Around a vertex,
consecutive incident edge ends delimit \emph{angles}, each
incident with a face. The degree counts edge ends, with a
loop contributing two. Euler's formula for a plane graph
with $n$ vertices, $m$ edges, $f$ faces, and $c$ connected
components is
\[
n-m+f=1+c.
\]

The \emph{directed dual} of a connected plane digraph has one
vertex for each face and one dual arc for each original arc.
The dual arc crosses the original arc once and is directed
from its right-hand face to its left-hand face, as viewed
along the original arc. Each original vertex corresponds to
a dual face whose boundary consists of the dual arcs associated
with its incident arcs~\cite{Erickson2020}.

\subsection{Directed Cuts}
\label{sec:cuts}
For a nonempty proper vertex set $U$ in a digraph $J$, the
arcs with exactly one endpoint in $U$ form a cut. It is a
\emph{directed cut} if it is nonempty and all its arcs point
in the same direction. The two sides are its \emph{shores}.
The side from which the arcs leave is its \emph{source shore}.
No arc enters this shore. In a weakly connected digraph,
every nonempty proper vertex set with no entering arc is
the source shore of a directed cut.

A weakly connected digraph is strongly connected if and only
if it has no directed cut. Indeed, a source shore cannot be
reached from its complement. Conversely, if the digraph is
not strongly connected, a source component of its condensation
has no entering arc and, by weak connectivity, has a leaving
arc, giving a directed cut.

Planar cycle--cut duality relates directed cuts to directed
cycles: a directed cycle in a connected plane digraph
corresponds to a directed cut in its dual, and the original
arcs corresponding to a directed cut in the dual contain a
directed cycle~\cite{Erickson2020,BorradaileKlein2009}. Geometrically, the cycle
separates the plane into an inside and an outside, and the
dual arcs crossing it all point from one side to the other.

An arc set \emph{meets} a directed cut if it contains an arc
of that cut. A family of $r$ pairwise arc-disjoint directed
cuts forces every arc set meeting all directed cuts to have
size at least $r$. The Lucchesi--Younger theorem states that
this lower bound is attained~\cite{LucchesiYounger1978}.

\begin{theorem}[Lucchesi--Younger]\label{thm:ly}
In a digraph, the minimum size of an arc set meeting every
directed cut equals the maximum number of pairwise arc-disjoint
directed cuts.
\end{theorem}
In Section~\ref{sec:cuttree-proof}, we apply this theorem to
certify how many reverse arcs are needed to make a weakly
connected digraph strongly connected.

\section{Reducing the Primal Graph}
\label{sec:primal}
This section develops the first stage of our algorithm:
simplifying the input digraph through structural reductions.

We first introduce the terminology used in the structural bounds.
Fix a planar embedding of the digraph. A \emph{directed face} is a face whose
boundary is a nonempty closed walk following the arc directions.
Such a boundary contains a directed cycle.
For a vertex $v$, the \emph{alternation count} $\sigma(v)$
is the number of changes between incoming and outgoing arcs
as we traverse the cyclic order of arcs around $v$.
We call $v$ \emph{ordinary} if $\sigma(v)=2$, in which case
its incoming arcs occur consecutively, as do its outgoing arcs.
We call $v$ \emph{special} if $\sigma(v)>2$.

We introduce five reduction rules and use planar counting
arguments to show that, after exhaustive application of these
rules, the number of directed faces, the number of special
vertices, and the total alternation count over special vertices
are each $O(k^3)$. These bounds provide the structural
information needed to compress the dual graph in the second
stage of our algorithm.

Two instances are \emph{equivalent} if they are both
YES-instances or both NO-instances. During the reductions,
$(D,k)$ denotes the current instance. A forced vertex deletion
decreases $k$ by one. If we determine the answer directly,
we output a fixed equivalent instance. We measure the encoding
size using explicit vertex and arc lists and a binary parameter.

A reduction rule is \emph{safe} if it produces an equivalent
instance or directly returns the correct answer. We apply
the rules in the order presented, restarting from the first
rule after each modification, until no rule applies.

\subsection{Elementary Reductions}
We first remove loops, repeated arcs, and vertices and arcs
that do not participate in directed cycles. We then eliminate
vertices with only one incoming or outgoing neighbor. The
resulting degree conditions will be essential in the next
subsection.

\begin{ruleenv}\label{rule:clean}
\leavevmode
\begin{enumerate}[label=(\roman*),leftmargin=*]
\item Delete each vertex with a loop, setting $k\gets k-1$
      for each deletion.
\item Remove repeated arcs with the same direction, arcs between
      distinct strongly connected components, and vertices not
      on any directed cycle.
\item Return NO if $k<0$. Otherwise, return YES if the graph
      is acyclic or $k\geq |V(D)|$.
\item Otherwise, the graph contains at least one directed cycle. If $k=0$, return NO. If $k=1$, check whether there
      exists a vertex whose deletion makes the graph acyclic.
      Return YES if such a vertex exists, and NO otherwise.
\end{enumerate}
\end{ruleenv}
\begin{lemma}\label{lem:clean-safe}
\Cref{rule:clean} is safe and preserves planarity.
\end{lemma}
\begin{proof}
A vertex with a loop belongs to every solution, so deleting
it and decreasing $k$ by one preserves equivalence. Repeated
copies of an arc do not change which vertex sets contain
directed cycles. Arcs between distinct strongly connected
components and vertices on no directed cycle can be removed
without affecting any solution. The answers in (iii) and (iv)
follow directly from the budget and acyclicity conditions.
All graph modifications are deletions and preserve planarity.
\end{proof}

Only repeated copies with the \emph{same direction} are removed.
Opposite arcs $x\to y$ and $y\to x$ are both retained: together
they form a directed cycle of length two. This distinction will
matter when we count edges and faces in the underlying graph.
%If the rule does not terminate the algorithm, then $k\geq2$.

\begin{ruleenv}\label{rule:degree}
If $v$ has a unique incoming neighbor $u$, remove $v$ and
replace each arc $v\to w$ by $u\to w$.
If $v$ has a unique outgoing neighbor $u$, remove $v$ and
replace each arc $w\to v$ by $w\to u$.
In either case, leave $k$ unchanged.
\end{ruleenv}

\begin{lemma}\label{lem:degree}
\Cref{rule:degree} is safe and preserves planarity.
\end{lemma}
\begin{proof}
Suppose $u$ is the unique incoming neighbor of $v$.
Every cycle through $v$ contains $u$, so replacing $v$ by
$u$ in a solution cannot increase its size. We may therefore
restrict attention to solutions avoiding $v$.

For a deletion set avoiding $v$, if $u$ is deleted, then
$v$ lies on no cycle, so removing $v$ does not affect acyclicity.
Otherwise, cycles correspond by
shortening $u\to v\to w$ to $u\to w$ and expanding new arcs
back through $v$, including a new loop at $u$.
Thus the two graphs are acyclic simultaneously, and their
minimum deletion sizes agree. Merging $v$ into $u$ along the
embedded edge $uv$ gives the reduced graph, so planarity is
preserved. The outgoing case is symmetric.
\end{proof}

Whenever we proceed beyond these two rules, every vertex has
at least two distinct incoming neighbors and at least two
distinct outgoing neighbors. Moreover, each weakly connected
component of the current instance is strongly connected,
because all arcs between distinct strongly connected components
have been removed. After any later modification, we apply
these rules again before using either property.

\subsection{Components after Two Vertex Deletions}
We next bound the number of weakly connected components that
can remain after two vertices are deleted from one strongly
connected component. This is the component count needed for
the planar face argument in the next subsection.

Fix a strongly connected component $D$ of the current instance,
and let $G$ be its underlying undirected graph, obtained by
replacing each arc by an undirected edge and keeping distinct arcs as
distinct edges. A weakly connected component may itself be
strongly connected, and isolated vertices also count as
weakly connected components. Since deleting vertices
commutes with taking the underlying undirected graph, for any two
distinct vertices $x,y$,
\[
\comp(G-\{x,y\})
=\text{the number of weakly connected components of }D-\{x,y\}.
\]

The following lemma uses the degree conditions established
above. It shows that each remaining weakly connected component
contains a cycle after either one of $x,y$ is restored. This
will force any solution avoiding that restored vertex to use
at least one deletion in each component.

\begin{lemma}\label{lem:components}
Assume that \Cref{rule:clean,rule:degree} are no longer
applicable. Let $D$ be a strongly connected component of
the current instance, and let $x,y\in V(D)$ be distinct.
For every weakly connected component $C$ of $D-\{x,y\}$,
both $D[C\cup\{x\}]$ and $D[C\cup\{y\}]$ contain a
directed cycle.
\end{lemma}
\begin{proof}
If $D[C]$ contains a directed cycle, that cycle proves both
claims. Otherwise $D[C]$ is acyclic. Choose a source $a$ of
$D[C]$ and a sink $b$ of $D[C]$ reachable from $a$, allowing
$a=b$.

The vertex $a$ has no incoming neighbor in $C$. It also has
no incoming neighbor in any other weakly connected component
of $D-\{x,y\}$, since there are no arcs between such components.
All its incoming neighbors in $D$ must therefore belong to
$\{x,y\}$. As $a$ has at least two distinct incoming neighbors,
both $x\to a$ and $y\to a$ exist. The same argument at the
sink $b$ gives both $b\to x$ and $b\to y$.

The path $a\reach b$ in $D[C]$ consequently extends to the
directed cycles
\[
x\to a\reach b\to x
\qquad\text{and}\qquad
y\to a\reach b\to y.
\]
These lie in the two required induced subgraphs.
\end{proof}

\begin{ruleenv}\label{rule:components}
If there are distinct vertices $x,y$ in a strongly connected
component $D$ such that $D-\{x,y\}$ has more than $k$
weakly connected components, delete $x,y$ and set $k\gets k-2$.
\end{ruleenv}
\begin{lemma}\label{lem:components-safe}
\Cref{rule:components} is safe.
\end{lemma}
\begin{proof}
Let $C_1,\ldots,C_q$ be the weakly connected components of
$D-\{x,y\}$, where $q>k$. Suppose a solution $X$ avoids $x$.
For each $i$, \Cref{lem:components} supplies a directed cycle
contained in $C_i\cup\{x\}$. Since $x\notin X$, the solution
must contain a vertex of $C_i$. The sets $C_i$ are pairwise
disjoint, so $|X|\geq q>k$. Thus every solution of size at
most $k$ contains $x$. Applying the same argument with $y$
shows that every such solution also contains $y$.

Deleting $x,y$ from a solution leaves at most $k-2$ vertices
that make the remaining graph acyclic. Conversely, adding
$x,y$ to any solution of size at most $k-2$ for the reduced
instance gives a solution of size at most $k$ for the original
instance. Hence the rule is safe.
\end{proof}

When this rule is no longer applicable, every pair of distinct
vertices in a component satisfies
$\comp(G-\{x,y\})\leq k$.
%No bound on the number of strongly connected components after the deletion is needed. In particular,
%a singleton strongly connected component without a loop need not require any deletion.

\subsection{Directed Faces}
We now use the component bound to control directed faces.
We first bound how many faces can contain the same two vertices,
then derive bounds for the number incident with one vertex
and for their total number. The next subsection uses these
bounds to control the special vertices and their total
alternation count.

All face counts are taken separately in the induced embedding
of each strongly connected component, including its outer face.
Fix one such component $D$, and let $G$ be its underlying
plane graph as defined above. Although repeated arcs in the
same direction have been removed, opposite arcs remain as
two distinct edges of $G$. Thus $G$ is still a multigraph,
with at most two edges between any pair of vertices.

For distinct $x,y\in V(G)$, let $\mu(x,y)$ be the number of
distinct faces whose boundaries contain both vertices, and let
$m_{xy}$ be the number of edges directly joining them. We have
$m_{xy}\in\{0,1,2\}$. It equals two exactly when both opposite
arcs are present in $D$. Degrees and edge counts in $G$ include
these edges separately, as required by Euler's formula.

\begin{lemma}\label{lem:commonfaces}
After \Cref{rule:clean,rule:degree,rule:components} have been
exhaustively applied, every pair of distinct vertices $x,y$
in a component satisfies
\[
\mu(x,y)\leq \comp(G-\{x,y\})+m_{xy}\leq k+2.
\]
\end{lemma}
\begin{proof}
Write $c=\comp(G-\{x,y\})$. By \Cref{rule:components},
$c\leq k$, and by \Cref{rule:clean}, $m_{xy}\leq2$.
It remains to prove $\mu(x,y)\leq c+m_{xy}$.
The claim is immediate if $x,y$ have no common face, so
assume that $\mu(x,y)>0$.

Let $f_x,f_y$ be the numbers of distinct faces incident with
$x,y$, and let $d_x,d_y$ be their degrees in $G$.
There are exactly $d_x$ angles around $x$, and each face
incident with $x$ occupies at least one of them. Hence
$f_x\leq d_x$, and similarly $f_y\leq d_y$.
We count the decrease in the number of faces when $x,y$
and their incident edges are deleted in two ways.

\proofstep{Counting the faces that merge.}
Deleting a vertex merges all faces incident with it.
Because $x$ and $y$ share at least one face, all faces incident with either vertex merge
into a single face of $G-\{x,y\}$ when both vertices and their
incident edges are deleted.
There are $f_x+f_y-\mu(x,y)$ distinct faces in their union,
and all other faces are unchanged. The number of faces
therefore decreases by
\[
f_x+f_y-\mu(x,y)-1.
\]

\proofstep{Computing the same decrease using Euler's formula.}
Let $n,m,p$ be the numbers of vertices, edges, and faces of
$G$. Since $G$ is connected, $p=m-n+2$. Put
$G'=G-\{x,y\}$, and denote its corresponding counts by
$n',m',p'$. We remove two vertices and
$d_x+d_y-m_{xy}$ edges: each edge joining $x$ and $y$ was
counted twice in $d_x+d_y$ but must be deleted only once.
Thus
\[
n'=n-2,\qquad m'=m-d_x-d_y+m_{xy}.
\]
The graph $G'$ has $c$ connected components. Euler's formula
for a possibly disconnected plane graph gives
\[
p'=m'-n'+c+1
   =m-n-d_x-d_y+m_{xy}+c+3.
\]
Consequently,
\[
p-p'=d_x+d_y-m_{xy}-c-1.
\]
Equating this with the first count and rearranging yields
\[
\begin{aligned}
\mu(x,y)
&=c+m_{xy}-(d_x-f_x)-(d_y-f_y)\\
&\leq c+m_{xy}\leq k+2.
\end{aligned}
\]
This proves the claim.
\end{proof}

The bound applies to all faces, and hence also to directed
faces. We use it to show that a vertex incident with too many
directed faces must belong to every solution within the budget.

\begin{ruleenv}\label{rule:vertexfaces}
If a vertex $v$ appears on the boundaries of more than
$k(k+2)$ directed faces, delete $v$ and set $k\gets k-1$.
\end{ruleenv}
\begin{lemma}\label{lem:vertexfaces-safe}
\Cref{rule:vertexfaces} is safe.
\end{lemma}
\begin{proof}
The boundary of a directed face is a nonempty closed directed
walk and therefore contains a directed cycle. Every solution
must delete a vertex on that boundary. A solution avoiding
$v$ must consequently meet every directed face incident with
$v$ at another vertex. By \Cref{lem:commonfaces}, each other
vertex meets at most $k+2$ of these faces. A set of at most
$k$ vertices avoiding $v$ can therefore meet at most $k(k+2)$
of them. Under the condition of the rule this is insufficient,
so every solution within the budget contains $v$.
Deleting $v$ and decreasing $k$ by one preserves equivalence.
\end{proof}

Once this rule is exhausted, each vertex meets at most
$k(k+2)$ directed faces. This immediately bounds how many
directed faces a solution of size at most $k$ can meet in total.

Let $F$ denote the total number of directed faces, summed
over all strongly connected components of the current instance.
Each component is considered with its own induced embedding,
including its outer face.
\begin{ruleenv}\label{rule:allfaces}
If $F>k^2(k+2)$, return NO.
\end{ruleenv}
\begin{lemma}\label{lem:allfaces-safe}
\Cref{rule:allfaces} is safe.
\end{lemma}
\begin{proof}
Every solution must meet the boundary of every directed face.
By \Cref{rule:vertexfaces}, each chosen vertex meets at most
$k(k+2)$ directed faces. Thus at most $k$ chosen vertices can
meet at most $k^2(k+2)$ directed faces, even when the counts
are summed over all components.
\end{proof}
The remaining instance therefore satisfies
\begin{equation}\label{eq:faces}
F\leq k^2(k+2)=O(k^3).
\end{equation}

\subsection{Special Vertices and Their Alternation Counts}

We now use the bound on directed faces to control the number
of special vertices and their total alternation count.
The argument counts source angles in two ways, first around
vertices and then along face boundaries. Combining these
counts with Euler's formula relates the alternation counts
to the number of directed faces.

\begin{lemma}\label{lem:switches}
The number of special vertices and the sum of their alternation
counts $\sigma(v)$ are $O(k^3)$.
\end{lemma}
\begin{proof}
Fix a strongly connected component $D$. Let $n=|V(D)|$,
$m=|A(D)|$, and let $\mathcal F_D$ be its set of faces,
including the outer face. Write $p=|\mathcal F_D|$ and let
$F_D$ be the number of directed faces in $\mathcal F_D$.
Each arc gives one edge of the connected underlying plane
multigraph, so Euler's formula is
\[
n-m+p=2.
\]

An \emph{angle} is a sector between consecutive incident
arcs in the cyclic order around a vertex. It is a
\emph{source angle} if both arcs leave the vertex. Every
angle is incident with exactly one face. When a face boundary
visits a vertex more than once, the corresponding angle
occurrences are counted separately.
Figure~\ref{fig:source-angles} illustrates the
two local configurations.

\begin{figure}[htbp]
\centering
\begin{tikzpicture}
  \begin{scope}[xshift=-2.6cm]
    \node[v] (vo) at (0,0) {$v$};
    \draw[arr,blue!70!black] (vo) -- ++(-1.05,.65);
    \draw[arr,blue!70!black] (vo) -- ++(1.05,.65);
    \draw[arr] (-1.05,-.65) -- (vo);
    \draw[arr] (1.05,-.65) -- (vo);
    \draw[blue!70!black] (-.36,.30)
      to[bend left=48] (.36,.30);
    \node[blue!70!black] at (0,1.05) {source angle};
    \node at (0,-1.05) {ordinary: $\sigma(v)=2$};
  \end{scope}

  \begin{scope}[xshift=2.6cm]
    \node[v] (vs) at (0,0) {$v$};
    \draw[arr] (-1.05,.65) -- (vs);
    \draw[arr,blue!70!black] (vs) -- ++(1.05,.65);
    \draw[arr,blue!70!black] (vs) -- ++(-1.05,-.65);
    \draw[arr] (1.05,-.65) -- (vs);
    \node at (0,1.05) {no source angle};
    \node at (0,-1.05) {special: $\sigma(v)=4$};
  \end{scope}
\end{tikzpicture}
\caption{Two cyclic orders of two incoming and
two outgoing arcs. The alternation count records changes of
direction, whereas a source angle lies between consecutive
outgoing arcs.}
\label{fig:source-angles}
\end{figure}
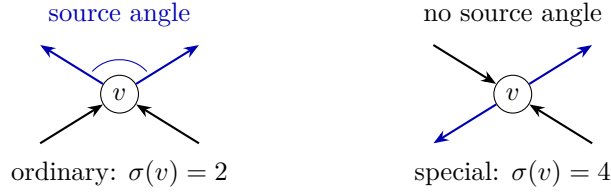

\proofstep{Counting source angles at a vertex.}
Fix $v\in V(D)$.
The vertex has both incoming and outgoing arcs. In their
cyclic order, the maximal runs of outgoing arcs alternate
with the maximal runs of incoming arcs. Each outgoing run
has one change from incoming to outgoing at its start and
one change back to incoming at its end. The number of outgoing
runs is therefore
\[
r_v=\frac{\sigma(v)}2.
\]
Let the lengths of these runs be $\ell_1,\ldots,\ell_{r_v}$,
so that $\sum_{i=1}^{r_v}\ell_i=d^+(v)$. A run of length
$\ell_i$ contributes exactly $\ell_i-1$ source angles, one
between each pair of consecutive outgoing arcs in that run.
The angles at either end of the run have one incoming and
one outgoing arc and are not source angles. Thus the number
$s(v)$ of source angles at $v$ is
\[
s(v)=\sum_{i=1}^{r_v}(\ell_i-1)
    =d^+(v)-r_v
    =d^+(v)-\frac{\sigma(v)}2.
\]
Summing over all vertices and using
$\sum_v d^+(v)=m$ gives the total number of source angles:
\[
\sum_{v\in V(D)}s(v)
=m-\sum_{v\in V(D)}\frac{\sigma(v)}2.
\]

\proofstep{Counting the same angles along face boundaries.}
For $f\in\mathcal F_D$, let $c(f)$ be the number of source-angle
occurrences on its boundary. Since each source angle is incident
with exactly one face, counting by faces or by vertices gives
the same total:
\[
\sum_{f\in\mathcal F_D}c(f)
=\sum_{v\in V(D)}s(v)
=m-\sum_{v\in V(D)}\frac{\sigma(v)}2.
\]
Rearranging, we obtain
\[
\sum_{v\in V(D)}\frac{\sigma(v)}2
+\sum_{f\in\mathcal F_D}c(f)=m.
\]

\proofstep{Applying Euler's formula.}
Subtract one for each of the $n$ vertices and each of the
$p$ faces. The preceding identity then gives
\[
\begin{aligned}
\sum_{v\in V(D)}\left(\frac{\sigma(v)}2-1\right)
+\sum_{f\in\mathcal F_D}(c(f)-1)
&=\left(\sum_v\frac{\sigma(v)}2+\sum_f c(f)\right)-n-p\\
&=m-n-p\\
&=-2,
\end{aligned}
\]
where the final equality is exactly $n-m+p=2$ rearranged.

\proofstep{Separating directed faces and special vertices.}
If $f$ is directed, a traversal of its boundary always enters
each visited vertex along one arc and leaves along the next.
It therefore encounters no source angle, so $c(f)=0$ and
$c(f)-1=-1$.

If $f$ is not directed, choose a direction in which to traverse
its closed boundary walk. Some arcs are traversed with their
directions and some against them. Otherwise, the face would
be directed. In this cyclic sequence there must be a change
from traversing an arc against its direction to traversing
the next arc with its direction. At their common vertex both
arcs point outwards, so this occurrence is a source angle.
Hence $c(f)\geq1$ and $c(f)-1\geq0$.

At an ordinary vertex, $\sigma(v)=2$, so the vertex term is
zero. Every remaining vertex has both incoming and outgoing
arcs, and hence an even alternation count of at least two.
Thus every nonordinary vertex is special. Removing the zero
terms for ordinary vertices and moving the $F_D$ contributions
of $-1$ from directed faces to the right gives
\begin{equation}\label{eq:angles}
\sum_{\substack{v\in V(D)\\v\text{ special}}}
\left(\frac{\sigma(v)}2-1\right)
+\sum_{\substack{f\in\mathcal F_D\\f\text{ not directed}}}
(c(f)-1)=F_D-2.
\end{equation}

\proofstep{Deriving the two bounds.}
Every term on the left of \eqref{eq:angles} is nonnegative.
A special vertex has $\sigma(v)\geq4$ and contributes at
least one. If $S_D$ is the set of special vertices in $D$,
then
\[
|S_D|
\leq\sum_{v\in S_D}\left(\frac{\sigma(v)}2-1\right)
\leq F_D-2\leq F_D.
\]
For $\sigma(v)\geq4$, we also have
\[
\sigma(v)\leq4\left(\frac{\sigma(v)}2-1\right),
\]
because the right-hand side minus the left-hand side is
$\sigma(v)-4\geq0$. Consequently,
\[
\sum_{v\in S_D}\sigma(v)
\leq4\sum_{v\in S_D}\left(\frac{\sigma(v)}2-1\right)
\leq4(F_D-2)\leq4F_D.
\]
Finally, sum over the strongly connected components.
Since $\sum_D F_D=F$, the entire instance has at most $F$
special vertices, with total alternation count at most $4F$.
Together with \eqref{eq:faces}, this proves both $O(k^3)$ bounds.
\end{proof}

Each reduction can be checked and applied in polynomial time,
and every modification decreases the number of vertices or
arcs. We therefore obtain in polynomial time an equivalent
instance with $O(k^3)$ directed faces, $O(k^3)$ special vertices,
and total alternation count $O(k^3)$ at those vertices.
The next section uses these bounds to construct the dual instance.

\section{The Dual Augmentation Problem}
\label{sec:dual}
This section turns the reduced vertex-deletion instance into
an equivalent instance of GSCA. We first use planar duality,
then reduce the sizes of the arc groups and identify a small
retained vertex set. Section~\ref{sec:compression} uses this
set to compress the dual instance.

\subsection{From Deletion to Augmentation}
This subsection uses planar duality to turn vertex deletions
that make the original graph acyclic into the addition of
groups of reverse arcs that make the dual strongly connected.

We now convert the reduced \PDFVS{} instance $(D,k)$ into a
GSCA instance $(H,\mathcal G,k)$. Here and in
Section~\ref{sec:compression}, $k$ denotes the remaining budget
after the primal reductions. First suppose that $D$ is strongly connected.
Construct its directed dual $H$, and call the arcs initially
present in $H$ its \emph{base arcs}. An original vertex $u$
corresponds to a dual face $F_u$, and each original arc incident
with $u$ becomes a boundary arc of $F_u$ in the dual.

For each $u\in V(D)$, let $R_u$ be the group containing the
reversals of all boundary arcs of $F_u$. This gives the group
family $\mathcal G=(R_u)_{u\in V(D)}$, with $|V(D)|$ groups.
Deleting $u$ corresponds to adding $R_u$, so the parameter
remains $k$. For multiple strongly connected components, apply
the construction separately and take the union of the resulting
graphs and group families, with the same overall parameter $k$.
The strong-connectivity arguments below concern a single
strongly connected component $D$ of the original graph and
its dual $H$.

Figure~\ref{fig:dual} illustrates this correspondence.
The two inside faces $f_1,f_2$ and the outside face $g$ become
three dual vertices. The reversals of the dual arcs corresponding
to the three original arcs incident with $u$ form
\[
R_u=\{g\to f_1,\ g\to f_2,\ f_1\to f_2\}.
\]
Deleting $u$ leaves only the directed path $v\to w\to z$ in
the original graph, whereas adding $R_u$ to the dual makes
its three vertices mutually reachable.

\begin{figure}[ht]
\centering
\begin{tikzpicture}
\node[v] (u) at (-1.4,1.2) {$u$};
\node[v] (v) at (1.4,1.2) {$v$};
\node[v] (w) at (1.4,-1.2) {$w$};
\node[v] (z) at (-1.4,-1.2) {$z$};
\draw[arr] (u)--(v);
\draw[arr] (v)--(w);
\draw[arr] (w)--(z);
\draw[arr] (z)--(u);
\draw[arr] (u)--(w);
\node at (.5,.6) {$f_1$};
\node at (-.6,-.5) {$f_2$};
\node at (2.1,.1) {$g$};
\node at (0,1.9) {original graph $D$};
\node at (0,-1.9) {$D-u$ is the path $v\to w\to z$};

\node[v] (f1) at (5.1,1.2) {$f_1$};
\node[v] (f2) at (5.1,-1.2) {$f_2$};
\node[v] (g) at (8.7,0) {$g$};
\draw[arr] (f1)--node[above] {$uv$}(g);
\draw[arr] (f1) to[bend right=10] (g);
\draw[arr] (f2)--node[below] {$zu$}(g);
\draw[arr] (f2) to[bend left=10] (g);
\draw[arr] (f2)--node[right] {$uw$}(f1);
\draw[arr,dashed] (g) to[bend right=20] (f1);
\draw[arr,dashed] (g) to[bend left=20] (f2);
\draw[arr,dashed] (f1) to[bend right=25] (f2);
\node at (3.9,0) {$F_u$};
\node at (6.9,1.9) {dual graph $H$};
\node at (6.9,-1.9) {$H+R_u$ is strongly connected};
\end{tikzpicture}
\caption{Solid arcs are base arcs. Dashed arcs form the group
$R_u$. The dual arcs labeled $uv,zu,uw$ correspond to the
original arcs with the same names and bound the dual face $F_u$.}
\label{fig:dual}
\end{figure}
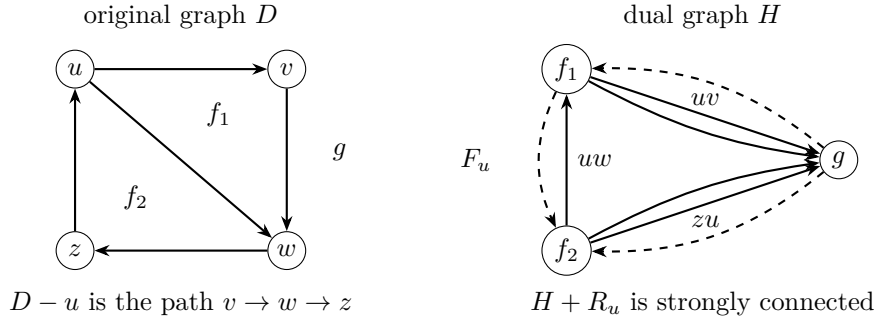

\begin{lemma}\label{lem:dual}
For every $X\subseteq V(D)$, the graph $D-X$ is acyclic if and only if adding to $H$ the groups associated with $X$ makes $H$ strongly connected.
\end{lemma}
\begin{proof}
Deleting $X$ removes all arcs with at least one endpoint in $X$.
Thus $D-X$ is acyclic exactly when every directed cycle
contains such an arc. By cycle--cut duality, this is equivalent
to every directed cut of $H$ containing a corresponding base arc.

Adding the groups associated with $X$ adds exactly the
reversals of these base arcs. If a directed cut contains such
a base arc, adding its reversal supplies arcs in both directions
across the cut, so it is no longer directed. Otherwise, the cut
still has arcs in only one direction. Adding arcs also cannot
turn a cut that already has arcs in both directions into a
directed cut.

Consequently, after these groups are added, $H$ has no directed
cut exactly when $D-X$ is acyclic. Since $H$ is weakly connected,
having no directed cut is equivalent to strong connectivity.
\end{proof}

\subsection{Reducing Group Sizes}
Each original arc incident with $u$ contributes the reversal
of its dual arc to $R_u$. Thus, if $u$ has large degree, $R_u$
may also be large. We reduce $R_u$ to at most $\sigma(u)$
arcs, using the vertex's alternation count to control the
group size.

The basic idea is to replace the reversals of all arcs on a
directed boundary path from $s$ to $t$ by one reverse arc
$t\to s$. Here and below, a reverse arc may therefore close
an entire base path. It need not be the reverse of a single
base arc. Adding this one arc makes all vertices on the path
mutually reachable, as illustrated in Figure~\ref{fig:segments}.
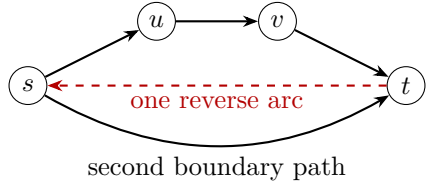
\begin{figure}[ht]
\centering
\begin{tikzpicture}
\node[v] (s) at (0,0) {$s$};\node[v] (t) at (5,0) {$t$};
\node[v] (u) at (1.7,.85) {$u$};\node[v] (v) at (3.3,.85) {$v$};
\draw[arr] (s)--(u);\draw[arr] (u)--(v);\draw[arr] (v)--(t);
\draw[arr] (s) to[bend right=30] node[below]{second boundary path}(t);
\draw[back] (t)--node[below]{one reverse arc}(s);
\end{tikzpicture}
\caption{The two boundary paths have the same start $s$ and
end $t$, so they can share the reverse arc $t\to s$.}
\label{fig:segments}
\end{figure}

\begin{ruleenv}\label{rule:boundarypaths}
For each group $R_u$, perform the following operations.
\begin{enumerate}[label=(\roman*)]
\item Split the boundary of the corresponding face $F_u$ at
every change of direction into directed paths.
\item For each path from $s$ to $t$, replace all its reverse
arcs in $R_u$ by one arc $t\to s$. If repeated arcs are
produced, keep only one copy.
\end{enumerate}
Leave $k$ unchanged.
\end{ruleenv}

\begin{lemma}\label{lem:segments}
\Cref{rule:boundarypaths} is safe.
\end{lemma}
\begin{proof}
The original component $D$ is strongly connected and hence
has no directed cut. By cycle--cut duality, the dual graph $H$
is acyclic. Every face boundary therefore has direction changes
and can be split into directed paths as in (i).

Fix the selected groups, and consider a directed path
$s\reach t$ on the face boundary corresponding to one of them.
For any base arc $u\to v$ on this path, after the replacement
we can return from $v$ to $u$ along
\[
v\reach t\to s\reach u.
\]
This path can therefore replace the reverse arc $v\to u$
that was added before the reduction.

Conversely, before the replacement, the reversals of all arcs
on the boundary path were added. Following these reverse arcs
gives a path from $t$ to $s$, which can replace the new arc
$t\to s$.

Thus, for the same choice of groups, reachability is the same
before and after the replacement, and the number of selected
groups is unchanged. Hence the rule is safe.
\end{proof}

We call the groups corresponding to ordinary and special
vertices \emph{ordinary groups} and \emph{special groups},
respectively.

\begin{lemma}\label{lem:group-properties}
After applying \Cref{rule:boundarypaths}, the instance has
the following properties.
\begin{enumerate}[label=(\roman*)]
\item Every ordinary group contains one arc. The group associated
with a special vertex $v$ contains at most $\sigma(v)$ arcs,
and the total number of arcs in special groups is $O(k^3)$.
\item For every base arc $u\to v$, at least one group contains
a reverse arc $t\to s$ whose corresponding directed boundary
path is $s\reach u\to v\reach t$. Adding all groups makes $H$
strongly connected.
\end{enumerate}
\end{lemma}
\begin{proof}
Direction changes along the face boundary correspond to changes
between incoming and outgoing arcs around the original vertex
$v$, so the boundary is split into $\sigma(v)$ directed paths.
For an ordinary vertex, the two paths have the same start and
end, producing the same reverse arc, so only one arc remains
in the group. For a special vertex, each path produces at most
one reverse arc, so the group contains at most $\sigma(v)$ arcs.
By \Cref{lem:switches}, the total number of arcs in special
groups is $O(k^3)$, proving (i).

Every base arc $u\to v$ belongs to a directed path
$s\reach u\to v\reach t$ obtained by the split, and the
corresponding reverse arc $t\to s$ is retained in some group.
Adding this reverse arc makes the vertices of the path mutually
reachable. Thus, after all groups are added, the endpoints of
every base arc are mutually reachable. Since $H$ is weakly
connected, it is then strongly connected, proving (ii).
\end{proof}

\subsection{Retained Vertices}
The previous subsection reduced the group sizes, but $H$ may
still contain many vertices.
To reduce the size of $H$, our algorithm merges vertices outside a carefully chosen subset. Before defining this merge operation, we first identify a set of retained vertices that will be kept intact during the compression.

In this subsection, $H$ denotes the disjoint union of the dual
graphs of all remaining original components. Let $W$ consist
of all sources and sinks of $H$,
together with all endpoints of reverse arcs in special groups.
We call these vertices \emph{retained vertices}. The compression
keeps these vertices unchanged and only merges vertices outside
$W$.

\begin{lemma}\label{lem:retained}
The size of $W$ and the total number of arcs in special
groups are both $O(k^3)$. Moreover, for every vertex $v$, there
is a path from a retained vertex to $v$ and a path from $v$ to
a retained vertex, both using base arcs alone.
\end{lemma}
\begin{proof}
Directed faces of the original graph correspond exactly to
sources and sinks of the dual graph. By \eqref{eq:faces},
there are $F=O(k^3)$ such vertices in total.
By \Cref{lem:group-properties}(i), the total number of arcs
in special groups is $O(k^3)$.
Each arc contributes at most two endpoints, giving
$|W|=O(k^3)$.

Since $H$ is acyclic, every vertex is reachable from
a source and reaches a sink. These paths use only base arcs,
and all sources and sinks belong to $W$.
\end{proof}
The next section uses distances from $W$ to merge some vertices
outside $W$. The main difficulty is to prove that this merging
preserves equivalence.

\section{Distance-Based Compression}
\label{sec:compression}
Our goal is now to compress the dual instance.
Although we have found a small retained set $W$, the graph
may still contain arbitrarily many vertices outside $W$.
We reduce their number by merging vertices whose distances
from $W$ agree up to a suitable threshold.
The main challenge is to prove that this operation is safe.

The proof is guided by directed-cut lower bounds on the cost
of achieving strong connectivity. A suitable family of directed
cuts gives separate requirements for adding reverse arcs:
every cut needs an entering arc, and each
individual reverse arc can satisfy at most one
requirement. We will represent these
cuts by the edges of a tree and then change which vertices
lie on each side. The distances from $W$ will ensure that
vertices to be merged lie on the same side of every cut,
while the separate entering-arc requirements are preserved.

We develop this argument in four steps.
Section~\ref{sec:distance-rule} defines the distances used for
merging, states the reduction rule, and bounds the number of
vertices that remain. Section~\ref{sec:cuttree-proof} proves
the general cut-tree representation for reverse-arc augmentation
in an arbitrary weakly connected digraph.
Section~\ref{sec:move-proof} uses the distances to move vertices
to compatible tree positions.
Finally, Section~\ref{sec:safety-overview} applies these two
structural results to prove the safety of the reduction.

\subsection{Distance Records and Their Number}
\label{sec:distance-rule}
We first specify which vertices may be merged and how the
arcs and groups change under this operation. We then prove
that only polynomially many vertices remain. The definition
uses distances from the retained set $W$. Planarity supplies
the bound on the number of distinct distance records.

We work separately in each weakly connected component. Write $H$
for its base digraph and $W\subseteq V(H)$ for its retained set.
The arcs of $H$ are the base arcs. Selectable reverse arcs are
specified by the groups and are added only when a group is chosen.
The graph $H$ is acyclic. By \Cref{lem:retained}, $W$ contains
all sources and sinks and all endpoints of special-group arcs,
and every vertex is reachable from, and can reach, a vertex of
$W$ using base arcs alone.

An \emph{available reverse arc} is an arc occurring in at least
one group, whether or not that group has been selected.
Let $\mathcal R$ be the set of all such arcs in this component.
By \Cref{rule:boundarypaths}, every available reverse arc
$v\to u$ has a corresponding base path $u\reach v$.
It may therefore reverse an entire path, rather than a single
base arc. We will distinguish the number of individual arcs
added from the number of groups selected: selecting one group
adds all of its arcs at unit cost.

Construct an auxiliary graph $L$ on $V(H)$ by taking all base arcs
and all reverse arcs in ordinary groups. Give base arcs length
zero and ordinary reverse arcs length one. No further reverse arcs
are added from special groups. Thus $L$ is fixed independently of any
chosen solution, and the shortest-path distance $d_L(a,v)$ counts
the minimum number of ordinary reverse arcs needed to reach $v$
from $a$. We use $d_L(a,v)=\infty$ when no such path exists.
The graph $L$ is planar: each ordinary reverse arc can be drawn
inside the dual face corresponding to its group.

Order the retained vertices as $w_1,\ldots,w_{|W|}$. For each
$v\notin W$, define its \emph{distance record} by
\[
\theta(v)=\bigl(\theta_1(v),\ldots,\theta_{|W|}(v)\bigr),
\qquad
\theta_i(v)=\min\{d_L(w_i,v),k+1\},
\]
where $\min\{\infty,k+1\}=k+1$. The record distinguishes every
distance at most $k$ and uses one value for all larger distances
and for nonreachability. Only vertices outside $W$ are identified.

These records define the \emph{identification classes}: vertices
outside $W$ belong to the same class exactly when their records
are equal, and each vertex of $W$ forms a singleton class.
Replace every class by one vertex, and replace the endpoints
of each arc by the vertices representing their classes.
We call the resulting graph the \emph{quotient graph}, denoted
by $\overline H$, and write
$\pi:V(H)\to V(\overline H)$ for the map sending a vertex to
its class. For $X\subseteq V(H)$, write
$\pi(X)=\{\pi(v):v\in X\}$; the image of an arc $u\to v$
is $\pi(u)\to\pi(v)$. Each group initially consists of the
images of its reverse arcs. The rule below removes redundant
arcs and groups; write $\overline{\mathcal G}$ for the resulting
indexed group family.

\begin{ruleenv}\label{rule:merge}
Form the quotient by the identification classes defined above,
preserving the grouping of reverse arcs, and normalize it as follows.
\begin{enumerate}[label=(\roman*)]
\item Delete base loops and keep one copy of each remaining base arc.
\item From every group, delete loops, arcs already present as
base arcs, and repeated copies of an arc. Discard empty groups.
\item Keep one ordinary group for each distinct remaining reverse
arc. Preserve the identities of the remaining special groups.
\end{enumerate}
Leave $k$ unchanged.
\end{ruleenv}

The normalization changes no reachability relation: loops and
repeated arcs add no reachability, arcs already present as base
arcs need not be selected, and empty groups are never needed.
Duplicate ordinary groups can be represented by a single group
at the same cost. Thus all vertices with a common record outside
$W$ are merged, while the retained vertices remain separate. The quotient need
not be planar. Although a record has $|W|$ entries with $k+2$
possible values each, planarity of $L$ gives a polynomial bound
on the number of records that actually occur.

For a digraph with nonnegative arc lengths, let $d(u,v)$
be the directed shortest-path distance, with value $\infty$ if
$v$ is unreachable from $u$. The incoming ball of radius $r$
centered at $v$ is $B^-(v,r)=\{u:d(u,v)\leq r\}$.
We use the following consequence of the directed-ball theorem
and the Sauer--Shelah lemma~\cite{LeWulffNilsen2024,Sauer1972}.

\begin{theorem}\label{thm:balls}
In a planar digraph with nonnegative arc lengths, fix a set $S$
of marked vertices. The number of distinct sets $B^-(v,r)\cap S$,
over all centers $v$ and radii $r\geq0$, is $O((|S|+1)^4)$.
\end{theorem}
\begin{proof}
Let $G$ denote the given digraph. An outgoing ball consists of
the vertices reachable from its
center within a given distance. Applied to planar digraphs,
the directed-ball theorem of Le and Wulff-Nilsen~\cite{LeWulffNilsen2024}
states that no set of five vertices has all of its subsets
realized as intersections with outgoing balls. The theorem
allows nonnegative arc lengths, including zero.
Reversing every arc preserves planarity and turns each incoming
ball into an outgoing ball, so the same statement holds for
incoming balls. The Sauer--Shelah lemma~\cite{Sauer1972} then gives
\[
\bigl|\{B^-(v,r)\cap S:v\in V(G),\ r\geq0\}\bigr|
\leq \sum_{i=0}^{\min\{4,|S|\}}\binom{|S|}{i}
=O((|S|+1)^4).
\]
\end{proof}

The next construction encodes an entire distance record as one
such intersection. For each retained vertex, marked leaves test
the distance thresholds $0,\ldots,k$ simultaneously.

\begin{lemma}\label{lem:profiles}
In a planar digraph with arc lengths zero or one, fix a vertex
set $W$. Form each vertex's distance record using distances
from the vertices of $W$, truncated at $k+1$. The number of distinct
records is
\[
O\bigl((|W|(k+1)+1)^4\bigr).
\]
\end{lemma}
\begin{proof}
Write $L$ for the given graph. As shown in Figure~\ref{fig:leaves},
for each $a\in W$, add new vertices $a_0,\ldots,a_k$,
each $a_j$ incident only with an arc directed to $a$ of length
$k-j$. These additions preserve planarity and do not change
distances between original vertices. Mark all $|W|(k+1)$ new
vertices.

\begin{figure}[htbp]
\centering
\begin{tikzpicture}
\node[v] (a) at (2.1,0) {$a$};\node[v] (v) at (5.2,0) {$v$};
\node[keep] (a0) at (0,1.15) {$a_0$};\node[keep] (aj) at (0,0) {$a_j$};\node[keep] (ak) at (0,-1.15) {$a_k$};
\draw[arr] (a0)--node[above] {$k$}(a);
\draw[arr] (aj)--node[above] {$k-j$}(a);
\draw[arr] (ak)--node[below] {$0$}(a);
\draw[arr] (a)--node[above] {$d_L(a,v)$}(v);
\node[align=center] at (3,-2.05) {A leaf $a_j$ belongs to the radius-$k$ incoming ball\\exactly when $d_L(a,v)\leq j$.};
\end{tikzpicture}
\caption{The distance from $a_j$ to $a$ is $k-j$. Thus $a_j$ can reach $v$ within length $k$ exactly when $a$ can reach $v$ within length $j$.}
\label{fig:leaves}
\end{figure}

For any original vertex $v$, consider the incoming ball
$B^-(v,k)$, the set of vertices that can reach $v$ within length
$k$. Starting from $a_j$, reaching $a$ uses length $k-j$,
leaving length $j$ for the rest of the path. Hence
\begin{equation}\label{eq:leaves}
a_j\in B^-(v,k)\quad\Longleftrightarrow\quad
(k-j)+d_L(a,v)\leq k\quad\Longleftrightarrow\quad d_L(a,v)\leq j.
\end{equation}

If the distance from $a$ to $v$ is $r\leq k$, then among the
marked vertices attached to $a$, exactly $a_r,a_{r+1},\ldots,a_k$
belong to the ball. The smallest index therefore gives the
distance $r$. If none belongs to the ball, the corresponding
vector entry is $k+1$.

The corresponding vector entry can be read this way for every
$a\in W$, so the marked vertices in the ball determine the
entire distance vector. Different distance vectors correspond
to different intersections.
There are $|W|(k+1)$ marked vertices, and \Cref{thm:balls} bounds
the number of intersections by $O((|W|(k+1)+1)^4)$, proving
the claim.
\end{proof}

By \Cref{lem:profiles,lem:retained}, the quotient satisfies
\begin{equation}\label{eq:quotientorder}
|V(\overline H)|
\leq |W|+O\bigl((|W|(k+1)+1)^4\bigr)=O(k^{16}).
\end{equation}
This establishes the size bound for the proposed identification.
To prove safety, we must show that the quotient cannot admit
a solution within the budget when the original instance has none.
We begin by describing the cut obstruction that such a solution
would have to overcome.

\subsection{Representing an Augmentation Obstruction by a Tree}
\label{sec:cuttree-proof}
We first establish a structural fact about arbitrary weakly
connected digraphs. The aim is to represent a lower bound on
the number of added reverse arcs by a tree: each tree edge
will represent one directed cut, and each reverse arc will
cross at most one such edge backwards. This representation
will allow us to adjust the cuts in the next subsection.
No group structure is needed for the result proved here.

Let $J$ be a weakly connected digraph. In this subsection,
we may add the reverse $v\to u$ of any arc $u\to v$ of $J$,
while retaining the original arc. Let $\rho(J)$ be the minimum
number of such additions needed to make $J$ strongly connected.
This minimum is finite, since adding the reverse of every arc
makes a weakly connected digraph strongly connected.
Recall that a \emph{source shore} is a nonempty proper vertex
set with no entering arc. Strong connectivity requires an
entering arc for every source shore of $J$.

\begin{lemma}\label{lem:cuttree}
Let $J$ be a weakly connected digraph, and let $r$ be an integer
with $1\leq r\leq\rho(J)$. There exist a directed tree $T$
with $r$ edges and a surjective map $f:V(J)\to V(T)$ such that:
\begin{enumerate}[label=(\roman*)]
\item For every arc $u\to v$ of $J$, either $f(u)=f(v)$ or
      $f(u)\to f(v)$ is an edge of $T$.
\item For every edge $e=x\to y$ of $T$, let $C_e$ be the
      component of the underlying undirected tree $T-e$
      containing $x$. Then $f^{-1}(C_e)$ is a source shore
      of $J$.
\item All vertices of any strongly connected component of $J$
      have the same image under $f$.
\end{enumerate}
\end{lemma}

Thus an arc of $J$ stays at one tree position or moves one
edge forwards, and adding its reverse moves at most one edge
backwards. It can therefore enter at most one of the $r$
source shores in (ii). The map assigns sets of vertices to
tree positions. These sets need not be strongly connected.
The tree need not be directed away from a single root.
Two shores $A,B$ \emph{cross} if all four sets
$A\cap B$, $A\setminus B$, $B\setminus A$, and
$V(J)\setminus(A\cup B)$ are nonempty. They are \emph{noncrossing} otherwise.
Figure~\ref{fig:shoretree} illustrates how tree edges represent
source shores that may overlap without crossing.

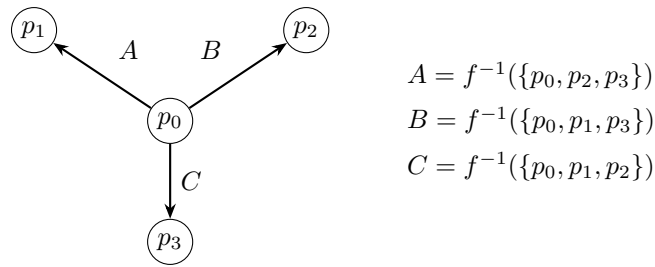
\begin{figure}[htbp]
\centering
\begin{tikzpicture}
\node[v] (p0) at (0,0) {$p_0$};
\node[v] (p1) at (-1.8,1.2) {$p_1$};
\node[v] (p2) at (1.8,1.2) {$p_2$};
\node[v] (p3) at (0,-1.6) {$p_3$};
\draw[arr] (p0)--node[midway,above right=3pt] {$A$}(p1);
\draw[arr] (p0)--node[midway,above left=3pt] {$B$}(p2);
\draw[arr] (p0)--node[right] {$C$}(p3);
\node[align=left,anchor=west] at (3,0)
{$A=f^{-1}(\{p_0,p_2,p_3\})$\\[5pt]
 $B=f^{-1}(\{p_0,p_1,p_3\})$\\[5pt]
 $C=f^{-1}(\{p_0,p_1,p_2\})$};
\end{tikzpicture}
\caption{Three overlapping, noncrossing source shores represented
by a tree. Deleting the edge labeled $A$ separates $p_1$ from
the positions whose inverse image is $A$; the other labels
have the same interpretation. A reverse step across one edge
enters only the shore represented by that edge.}
\label{fig:shoretree}
\end{figure}

\begin{proof}
\proofstep{Finding a family of cuts.}
A set of added reverse arcs makes $J$ strongly connected
if and only if the corresponding original arcs meet every
directed cut of $J$. Indeed, the reverse of an arc leaving
a source shore enters that shore. Every source shore of the
augmented graph would also be a source shore of $J$, so
supplying an entering arc to every original source shore
eliminates all directed cuts.

Apply \Cref{thm:ly} to $J$.
It follows that $\rho(J)$ is the maximum number of pairwise
arc-disjoint directed cuts. Choose $r$ such cuts, and let
$S_1,\ldots,S_r$ be their source shores. No arc of $J$ enters
any $S_i$, and each arc leaves at most one of them.

\proofstep{Making the shores noncrossing.}
Whenever two source shores $A,B$ cross, replace them by
$A\cap B$ and $A\cup B$. These are again nonempty
proper source shores.

Write $\mathbf1_S(x)=1$ if $x\in S$ and $0$ otherwise.
For any source shore $S$ and arc $u\to v$ of $J$, the
difference $\mathbf1_S(u)-\mathbf1_S(v)$ is zero or one,
and is one exactly when the arc leaves $S$. The identity
\[
\mathbf1_A(x)+\mathbf1_B(x)
=\mathbf1_{A\cap B}(x)+\mathbf1_{A\cup B}(x)
\]
at the two endpoints shows that the replacement preserves
the number of shores left by every arc. In particular,
each arc still leaves at most one shore.

Each replacement increases the sum of squared shore sizes
by $2|A\setminus B|\,|B\setminus A|>0$. This integer
potential is bounded by $r|V(J)|^2$, so the process ends
with $r$ pairwise noncrossing shores. They remain distinct:
weak connectivity ensures that every nonempty proper source
shore has a leaving arc, and a repeated shore would make
that arc leave two members of the family. Noncrossing does
not require the shores to be disjoint or nested. Two shores
may overlap and have union $V(J)$.

\proofstep{Constructing the tree.}
Assign each vertex $v$ its membership vector
\[
\bigl(\mathbf1_{S_1}(v),\ldots,\mathbf1_{S_r}(v)\bigr).
\]
Each distinct vector is a \emph{position}, and $f(v)$ is
the position of $v$. Form an undirected position graph by
joining $f(u)$ and $f(v)$ for each arc $u\to v$ of $J$,
omitting loops and repeated edges. This graph is connected
because $J$ is weakly connected.

An arc of $J$ enters no shore and leaves at most one.
Consequently, every position edge changes exactly one
membership coordinate. Every shore has a leaving arc,
so every shore coordinate is changed by some position edge.
There is only one position edge for each coordinate.
To see this, suppose two distinct edges both change the
coordinate for a shore $A$. Their endpoint vectors differ
in some other coordinate, corresponding to a shore $B$.
This coordinate is constant on each edge and has different
values on the two edges. The four endpoints therefore
realize all four membership combinations for $A$ and $B$,
contradicting that the shores are noncrossing.

The position graph thus has exactly one edge for each shore.
That edge is the only edge between the positions inside
the shore and those outside, so it is a bridge. A connected
graph in which every edge is a bridge is a tree. Denote this
tree by $T$. It has $r$ edges. Orient the edge corresponding
to $S_i$ from positions inside $S_i$ to positions outside it.
Every arc of $J$ joining distinct positions follows this
orientation, proving (i). The map $f$ is surjective because
every position is realized by a vertex.

\proofstep{Recovering the shores and strongly connected components.}
Deleting the edge corresponding to $S_i$ separates precisely
the positions inside $S_i$ from those outside it. The tail
component therefore has inverse image $S_i$, proving (ii).
Finally, a directed path of $J$ maps to a walk that follows
the directions of $T$, allowing steps that stay at a position.
Vertices in the same strongly connected component have
mutually reachable images. A directed tree has no directed
cycle, so these images must coincide. This proves (iii).
\end{proof}

The lemma gives a tree whose edges encode separate entering-arc
requirements. When we apply this representation to our component,
vertices with the same distance record may still occupy different
positions. The next subsection constructs a new map that respects
these records while retaining the properties needed to preserve
the cut lower bound.

\subsection{Making the Cuts Compatible with Identification}
\label{sec:move-proof}
We now return to the component $H$, the retained set $W$, and
the distances defined in Section~\ref{sec:distance-rule}.
A tree representation may place vertices with the same record
at different nodes. Our task is to move them to a common node,
while fixing the retained vertices and keeping the restriction
that each reverse arc moves at most one tree edge backwards.
This will allow the cuts to pass through the merging operation.

The following lemma states exactly what we need from the
initial tree representation. Section~\ref{sec:safety-overview}
will obtain such a representation from \Cref{lem:cuttree}.
Here we prove how to modify it using only the distance records.

\begin{lemma}\label{lem:move}
Let $H,W,\mathcal R,L$, and $\theta$ be as in
Section~\ref{sec:distance-rule}, with the structural properties
established in \Cref{lem:group-properties,lem:retained}.
In particular, every available reverse arc closes a base path,
and every base arc lies on a base path whose reverse belongs
to $\mathcal R$.
Let $T$ be a directed tree with at most $k+1$ edges, and let
$f:V(H)\to V(T)$ be a map with the following properties:
for every base arc $u\to v$ of $H$, either $f(u)=f(v)$ or
$f(u)\to f(v)$ is an edge of $T$; and for every reverse arc
$v\to u\in\mathcal R$, either $f(u)=f(v)$ or
$f(u)\to f(v)$ is an edge of $T$.
There is a map $v\mapsto m_v\in V(T)$ such that:
\begin{enumerate}[label=(\roman*)]
\item $m_w=f(w)$ for every $w\in W$.
\item If $u,v\notin W$ and $\theta(u)=\theta(v)$, then
      $m_u=m_v$.
\item For every base arc $u\to v$, either $m_u=m_v$ or
      $m_u\to m_v$ is an edge of $T$. For every reverse arc
      $v\to u\in\mathcal R$, either $m_v=m_u$ or
      $m_u\to m_v$ is an edge of $T$.
\end{enumerate}
\end{lemma}
\begin{proof}
\proofstep{Defining the allowed positions.}
For $x,y\in V(T)$, let $\delta(x,y)$ be the number of edges
traversed against their directions on the unique path from $x$
to $y$. In particular, $\delta(x,y)=0$ if and only if $y$ is
reachable from $x$ in the directed tree. Define the set of
\emph{allowed positions} for each $v\in V(H)$ by
\begin{equation}\label{eq:allowed}
A_v=\{x\in V(T):\delta(f(a),x)\leq d_L(a,v)
\text{ for every }a\in W\}.
\end{equation}
Thus a distance of $q$ in $L$ permits at most $q$ backward
tree edges from the corresponding retained position.

The original position $f(v)$ belongs to $A_v$. Indeed, an
$L$-path maps under $f$ to a tree walk in which base arcs move
forwards and each ordinary reverse arc contributes at most one
backward step. Removing detours cannot increase this count.
Applying this observation to a shortest $a$--$v$ path gives
$\delta(f(a),f(v))\leq d_L(a,v)$. If no such path exists,
the inequality is automatic. Moreover, $T$ has at most $k+1$
edges, so every value of $\delta$ is at most $k+1$.
Consequently, a distance $d_L(a,v)\geq k+1$, including an
infinite distance, imposes no restriction in \eqref{eq:allowed}.
The truncated record $\theta(v)$ therefore determines $A_v$
completely for $v\notin W$.

\proofstep{Choosing a common position from the distance record.}
For a fixed $a\in W$, the positions satisfying the corresponding
restriction in \eqref{eq:allowed} form a connected subtree.
Along the path from $f(a)$ to any allowed position, the number
of backward edges never exceeds its value at the endpoint.
The set is also closed under forward reachability: a forward
step cannot increase $\delta(f(a),\cdot)$. Since a nonempty
intersection of connected subtrees of a tree is connected,
$A_v$ is connected and closed under forward reachability.

There is a vertex $c\in W$ that reaches $v$ using only base
arcs: following base predecessors backwards from $v$ reaches
a source, and all sources belong to $W$. Hence $d_L(c,v)=0$,
and \eqref{eq:allowed} places the entire set $A_v$ in the part
of $T$ reachable forwards from $f(c)$. Rooted at $f(c)$, this
part is an outward-directed tree. Its nonempty connected
subtree $A_v$ has a unique node nearest the root. Otherwise,
the path between two distinct nearest nodes would contain a
node of $A_v$ closer to the root. Call the nearest node $m_v$.
Since $A_v$ is connected, the node
$m_v\in A_v$ nearest the root is an ancestor of every node of
$A_v$ and hence reaches all of them along the tree's directions.
This characterization is unique, since two distinct nodes of
a directed tree cannot be mutually reachable,
and is independent of the choice of $c$.

We call $m_v$ the \emph{earliest allowed position}: it is the
unique node in $A_v$ from which every other node in $A_v$ is
reachable. Figure~\ref{fig:positions} illustrates this choice.

\begin{figure}[htbp]
\centering
\begin{tikzpicture}
\node[keep] (s) at (0,0) {$s$};
\node[v,fill=green!18] (x) at (2,0) {$x$};
\node[keep] (t1) at (4,1) {$t_1$};
\node[keep] (t2) at (4,-1) {$t_2$};
\draw[arr] (s)--(x);
\draw[arr] (x)--(t1);
\draw[arr] (x)--(t2);
\node[align=left] at (7,0)
{$W=\{s,t_1,t_2\}$\\$\theta(v)=(0,1,1)$\\
$A_v=\{x,t_1,t_2\}$\\earliest position $m_v=x$};
\end{tikzpicture}
\caption{Allowed positions in a directed tree, with each retained vertex fixed at its original
position. Reaching $s$ from either
sink requires two backward edges, so $s$ is excluded. The other
three positions are allowed, and $x$ reaches both allowed sinks.}
\label{fig:positions}
\end{figure}
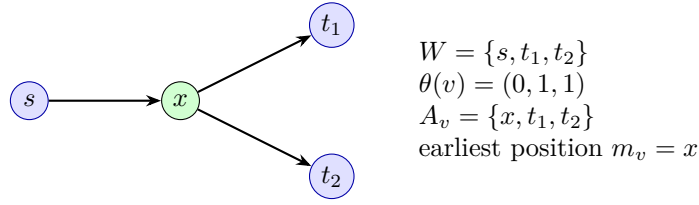

\proofstep{Retained vertices and equal records.}
If $v\in W$, the constraint in \eqref{eq:allowed} for $a=v$ has right-hand side $d_L(v,v)=0$.
Every position of $A_v$ is therefore reachable from $f(v)$,
and $f(v)$ itself is allowed. It follows that $m_v=f(v)$.
For vertices outside $W$, equal records give equal allowed
sets and hence equal earliest positions. Thus the assignment
$v\mapsto m_v$ is constant on every identification class and
fixes the positions of all retained vertices.

\proofstep{Base paths and special-group arcs.}
If $u\reach v$ using base arcs, then
$d_L(a,v)\leq d_L(a,u)$ for every $a\in W$.
The restrictions defining $A_v$ are therefore at least as
strong as those defining $A_u$, so $A_v\subseteq A_u$.
In particular, $m_v\in A_u$, and the choice of $m_u$ gives
$m_u\reach m_v$ in $T$. Thus base reachability is preserved
as forward reachability. Every reverse arc in a special group
has both endpoints in $W$, so their positions do not change.
The hypothesis on $f$ therefore implies that such an arc still
has equal endpoint positions or moves one step backwards.

\proofstep{The one-edge bound for ordinary reverse arcs.}
Let $v\to u$ be a reverse arc in an ordinary group. A base path
from $u$ to $v$ gives a directed tree path from $m_u$ to $m_v$.
The reverse arc has length one in $L$, and hence
\begin{equation}\label{eq:onearc}
d_L(a,u)\leq d_L(a,v)+1\qquad(a\in W).
\end{equation}
We show that the directed path from $m_u$ to $m_v$ has at
most one edge. This is immediate if $m_u=m_v$, so assume
otherwise. Write $\ell$ for its number of edges and
$z\to m_v$ for its last edge.

The predecessor $z$ does not belong to $A_v$. If it did,
the defining property of $m_v$ would require a directed
path from $m_v$ back to $z$, which is impossible across
the edge $z\to m_v$ in a tree. Thus there is some $a\in W$
whose constraint fails at $z$ and holds at $m_v$.
Putting $q=d_L(a,v)$, we have
\[
\delta(f(a),z)>q,
\qquad \delta(f(a),m_v)\leq q.
\]
The failed constraint implies that $q$ is finite.
Backward distances to adjacent tree nodes differ by at
most one, and they are integers. Therefore
\[
\delta(f(a),z)=q+1,
\qquad \delta(f(a),m_v)=q.
\]
In particular, this constraint is tight at $m_v$.

Delete the edge $z\to m_v$. The node $f(a)$ must lie in
the component containing $m_v$. Otherwise the path from
$f(a)$ to $m_v$ would first reach $z$ and then use the
forward edge $z\to m_v$, giving equal backward distances
to $z$ and $m_v$, a contradiction. Since $m_u$ lies in
the other component, the unique path from $f(a)$ to
$m_u$ first reaches $m_v$ and then traverses all $\ell$
edges of the directed $m_u$--$m_v$ path backwards.
Consequently,
$\delta(f(a),m_u)=q+\ell$.
Figure~\ref{fig:tightpath} displays this decomposition.
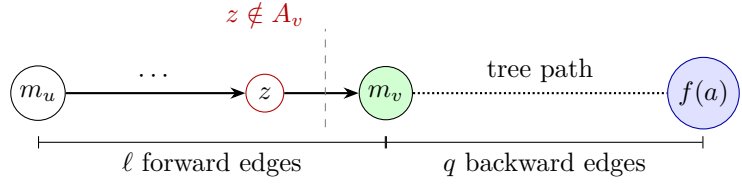
\begin{figure}[htbp]
\centering
\begin{tikzpicture}
\node[v] (mu) at (0,0) {$m_u$};
\node[v,draw=red!70!black] (z) at (3,0) {$z$};
\node[v,fill=green!18] (mv) at (4.6,0) {$m_v$};
\node[keep] (a) at (8.8,0) {$f(a)$};
\draw[arr] (mu)--node[above] {$\cdots$}(z);
\draw[arr] (z)--(mv);
\draw[thick,densely dotted] (mv)--node[above]{tree path}(a);
\draw[dashed,gray] (3.8,-.45)--(3.8,.85);
\draw[|-|] (0,-.65)--node[below] {$\ell$ forward edges}(4.6,-.65);
\draw[|-|] (4.6,-.65)--node[below] {$q$ backward edges}(8.8,-.65);
\node[red!70!black] at (3,1.05) {$z\notin A_v$};
\end{tikzpicture}
\caption{The excluded predecessor $z$ yields a retained
vertex $a$ for which the constraint at $m_v$ is tight.
The path from $f(a)$ to $m_v$ contains exactly
$q=d_L(a,v)$ backward edges; directions within this dotted
path may vary. Continuing to $m_u$ adds $\ell$ backward
edges. The depicted long directed path is the configuration
whose length the argument bounds.}
\label{fig:tightpath}
\end{figure}

Since $m_u\in A_u$, its constraint for this same vertex
$a$, together with \eqref{eq:onearc}, gives
\begin{equation}\label{eq:tight}
d_L(a,v)+\ell=\delta(f(a),m_u)
\leq d_L(a,u)\leq d_L(a,v)+1.
\end{equation}
It follows that $\ell\leq1$. Thus every reverse arc in an
ordinary group has equal endpoint positions or crosses one
tree edge backwards.
The earliest-position choice is essential here: exclusion
of its predecessor supplies the tight constraint needed
in \eqref{eq:tight}.

\proofstep{The one-edge bound for base arcs.}
For a base arc $u\to v$, \Cref{lem:group-properties}(ii) gives
an available reverse arc $z\to s$ and a base path
$s\reach u\to v\reach z$. Preservation of base reachability
yields $m_s\reach m_u\reach m_v\reach m_z$ in $T$.
Every available reverse arc belongs to an ordinary or a special
group, and we have shown in both cases that its endpoint
positions are equal or one backward step apart. In particular,
the directed path from $m_s$ to $m_z$ has at most one edge.
Its intermediate positions $m_u,m_v$ are therefore also equal
or forward-adjacent. This proves all the assertions of
\Cref{lem:move}.
\end{proof}

The new map respects all identification classes and preserves
the required restrictions on arc endpoints. To use it in the
safety proof, we must also ensure that moving vertices does
not empty either side of a cut. We establish this by finding
retained vertices on both sides before applying the lemma.

\subsection{Safety of the Reduction}
\label{sec:safety-overview}
We now prove that a solution on the quotient can be lifted
to the original component without increasing its cost.
There is one point to address before applying the structural
lemmas: \Cref{lem:cuttree} allows reverses of existing arcs,
whereas an available reverse arc of our instance may reverse
an entire base path. We first connect these two augmentation
models. We then construct the tree, find retained vertices
on both sides of every cut, and apply \Cref{lem:move}.
The resulting cuts give the contradiction needed for safety.

\begin{lemma}\label{lem:safety}
\Cref{rule:merge} preserves whether at most $k$ groups in total
can make every weakly connected component strongly connected.
\end{lemma}
\begin{proof}
We may ignore the normalization in the rule, since the deleted
arcs do not affect reachability and every retained group has an
original representative of the same cost.
The forward implication follows from preservation of paths under
identification. Duplicate ordinary groups can be replaced by
their single retained copy. For the converse, fix a solution
of the quotient using at most $k$ groups and consider one
component. Let $Y$ be its selected special groups, and let
$\mathcal O$ be its selected ordinary groups, with
$t=|\mathcal O|$. We use the original groups represented by
these quotient groups. Write $H+Y$ for the graph obtained
from $H$ by adding all reverse arcs belonging to groups in $Y$.

It suffices to show that $H+Y$ can be made strongly connected
by at most $t$ arcs from $\mathcal R$, counted individually,
since each such arc can be supplied by one
original group containing it.
Suppose, for a contradiction, that more than $t$ arcs from
$\mathcal R$ are necessary.

\proofstep{Relating the two augmentation models.}
For each reverse arc $v\to u\in\mathcal R$, add a distinguished
forward arc $u\to v$ to a set $P$, and define
\[
J=H+Y+P.
\]
Every arc of $P$ has a corresponding base path with the same
endpoints, so adding $P$ does not change reachability or the
source shores of $H+Y$. Parallel copies may be kept to
distinguish the origins of arcs. We claim that $\rho(J)>t$,
where $\rho$ is the reverse-arc augmentation cost defined in
Section~\ref{sec:cuttree-proof}.

To prove the claim, consider the reverse of any existing arc
of $J$. It can be simulated in $H+Y$ using at most one arc
from $\mathcal R$, as follows.
\begin{itemize}
\item For a base arc $u\to v$, \Cref{lem:group-properties}(ii)
      gives an arc $z\to s\in\mathcal R$ and a base path
      $s\reach u\to v\reach z$. Adding $z\to s$ supplies
      the walk $v\reach z\to s\reach u$, which simulates
      the reverse $v\to u$.
\item For an arc of $P$, its reverse belongs to $\mathcal R$
      by construction.
\item For an arc added by $Y$, its reverse can already be
      simulated by a base path, so no further arc is needed.
\end{itemize}
If at most $t$ reverses of existing arcs made $J$ strongly
connected, choose the at most $t$ arcs of $\mathcal R$ used
in these simulations. Every arc of the augmented $J$ would
then have a corresponding walk in $H+Y$ with those arcs added:
arcs of $P$ use their base paths, and the new reverses use
the walks above. The same vertices would be mutually reachable
there, contradicting our assumption. This proves $\rho(J)>t$.

\proofstep{Constructing the tree and anchoring the cuts.}
The graph $J$ is weakly connected, so \Cref{lem:cuttree}
applies with $r=t+1$. It gives a directed tree $T$ with
$t+1\leq k+1$ edges and a map $f:V(H)\to V(T)$.
Every base arc maps to one position or one forward tree edge.
For every $v\to u\in\mathcal R$, the corresponding arc
$u\to v\in P$ has the same property. Thus $T,f$ satisfy
the hypotheses of \Cref{lem:move}.

Two further properties will be needed after vertices move.
First, each reverse arc added by $Y$ has equal endpoint
images under $f$. Its reverse is a base path, so its endpoints
belong to the same strongly connected component of $J$;
\Cref{lem:cuttree}(iii) applies. Both endpoints are in $W$.

Second, each side of every tree edge contains the image of
a retained vertex. To see this, fix an edge $e=x\to y$ and
let $C_e$ be the component of the underlying tree $T-e$
containing $x$. By \Cref{lem:cuttree}(ii), the set
$S_e=f^{-1}(C_e)$ is a nonempty proper source shore of $J$.
Choose a vertex in $S_e$. By \Cref{lem:retained}, it is
reachable along a base path from some $a\in W$. Since no
base arc enters $S_e$, that path starts in $S_e$, so
$a\in S_e$. Likewise, a vertex outside $S_e$ can reach
some $b\in W$ along a base path. This path cannot enter
$S_e$, and hence $b\notin S_e$. These two retained vertices
anchor the two sides of the edge. We do not need a retained
vertex at every tree node.

\proofstep{Transferring the cuts to the quotient.}
Apply \Cref{lem:move} to obtain $v\mapsto m_v$.
For each tree edge $e=x\to y$, let $C_e$ again denote the
component of the underlying tree $T-e$ containing $x$, and define
\[
U_e=\{v\in V(H):m_v\in C_e\},
\qquad \overline U_e=\pi(U_e).
\]
The set $U_e$ is a union of complete identification classes:
vertices outside $W$ with the same record have the same new
position, and every retained vertex forms a singleton class.
Consequently $\pi^{-1}(\overline U_e)=U_e$.
The retained vertices just found stay on their respective
sides because $m_w=f(w)$ for $w\in W$. Thus both $U_e$
and $\overline U_e$ are nonempty proper vertex sets.

No base arc enters $\overline U_e$: its endpoint positions
are equal or joined by a forward tree edge, whereas entering
$C_e$ requires traversing $e$ backwards. No arc added by $Y$
enters $\overline U_e$ either. Its endpoints had equal images
under $f$ and are retained, so their images remain equal.
Finally, each available reverse arc enters at most one of
the sets $\overline U_e$, since it crosses at most one tree
edge backwards.

The $t+1$ tree edges give $t+1$ entering-arc requirements,
all of which are necessary for strong connectivity. The $t$ groups
in $\mathcal O$ supply only $t$ arcs, because each ordinary
group contains one arc. At least one nonempty proper set
$\overline U_e$ therefore still has no entering arc after
all selected groups are added. This contradicts that
$Y\cup\mathcal O$ is a solution on the quotient.

\proofstep{Recovering a group solution.}
It follows that $H+Y$ can be made strongly connected by at most
$t$ arcs from $\mathcal R$.
For each such arc, select one original group
containing it, ordinary or special. This uses at most $t$
additional groups; repeated choices or groups already in $Y$
only reduce the number, and any other arcs added with a selected
group cannot destroy strong connectivity.
Apply this argument separately to every component and sum the
group counts. Groups are confined to individual
components, so the total is no larger than that of the
quotient solution and is therefore at most $k$.
\end{proof}

The reduction is constructive: compute the $0/1$ distances from
$W$, sort the records, and merge the vertices in each class.
The cut tree and the reassignment are used only in the safety
proof and need not be computed by the algorithm. The quotient
has polynomial size, but is an instance of the augmentation
problem and need not be planar. The next section converts it
into an instance of \PDFVS{}.

\section{Returning to Vertex Deletion}
\label{sec:return}
The preceding section gives a polynomial-size GSCA instance.
To obtain a kernel for PDFVS, we convert this instance back
to vertex deletion in three steps.
First, we encode the choice of groups and strong connectivity
by a $3$-CNF formula. Second, we make its incidence graph
planar, with a quadratic increase in formula size. Third, we
replace the planar formula by a planar digraph whose directed
triangles enforce a consistent satisfying assignment. We prove
each equivalence and then combine the size bounds.

We use $3$-CNF formulas whose clauses contain at most three
literals. A literal is a Boolean variable or its negation,
a clause is a disjunction of literals, and a CNF formula is
a conjunction of clauses. Formula size is measured by the
total number of literal occurrences. Two formulas are
\emph{equisatisfiable} if they are both satisfiable or both
unsatisfiable.

\subsection{Reduction to \texorpdfstring{$3$}{3}-SAT}
\label{sec:formula-tools}
We first describe a Boolean circuit that checks a certificate
of strong connectivity, and then combine it with a circuit
encoding the choice of arc groups. This gives an equivalent
circuit for the GSCA instance, which we convert into a
$3$-CNF formula. We use the following standard encoding to
bound the number of literal occurrences in the formula.

A Boolean circuit is an acyclic network of logical gates with
a designated output. It accepts an input assignment when the
output is true. We use two-input AND gates and NOT gates.
OR gates can be expressed by a constant number of these gates.

\begin{lemma}\label{lem:circuit-cnf}
A Boolean circuit with $g$ gates, each a two-input AND gate
or a NOT gate, can be converted in $O(g)$ time into a
$3$-CNF formula with $O(g)$ literal occurrences.
An assignment to the circuit inputs extends to a satisfying
assignment of the formula if and only if the circuit accepts it.
\end{lemma}
\begin{proof}
Introduce a variable for each gate output.
Enforce $z=x\land y$ by the clauses
$(\neg z\lor x)$, $(\neg z\lor y)$, and
$(z\lor\neg x\lor\neg y)$.
Enforce $z=\neg x$ by $(\neg z\lor\neg x)$ and $(z\lor x)$.
A unit clause contains one literal, so $(z)$ fixes $z=1$
and $(\neg z)$ fixes $z=0$.
Use these clauses to fix constants and require the circuit
output to be true. Following the gates in dependency order,
the clauses force exactly the circuit evaluation, so satisfying
assignments correspond to accepting evaluations.
Each gate contributes only constantly many literal occurrences.
\end{proof}

\begin{lemma}\label{lem:cnf}
A GSCA instance with parameter $k\geq1$, at most $N\geq2$
vertices, $O(N^2)$ arc groups, and $O(N^2)$ base arcs and
group-arc entries in total can be converted in deterministic
polynomial time into a $3$-CNF formula with $O(kN^2\log N)$
literal occurrences. The formula is satisfiable if and only
if the instance has a solution.
\end{lemma}
\begin{proof}
The proof has three parts. We first construct a circuit
that checks a certificate of strong connectivity and prove
its correctness. We then combine this circuit with an
encoding of the group selection in a GSCA instance.
Finally, we bound the size of the resulting circuit and
apply \Cref{lem:circuit-cnf} to obtain a $3$-CNF formula
of the claimed size.

\proofstep{A circuit certifying strong connectivity.}
Let $V$ be a fixed set of $n\geq2$ vertices, let $A$ be a
set of potential arcs on $V$, and choose a root $r\in V$.
For each $u\to v\in A$, an input bit $a_{uv}$ indicates
whether the arc is present. Write $H_a$ for the resulting
digraph. We construct a circuit that accepts precisely when
two additional vertex labelings certify that $H_a$ is strongly
connected.

Give each vertex $v$ two nonnegative integer labels $x(v)$
and $y(v)$, each represented by $h=\lceil\log_2 n\rceil$
input bits. The labels certify reachability to and from the
root, respectively. Require $x(r)=y(r)=0$, and impose the
following conditions on every vertex $v\neq r$:
\begin{itemize}
\item The \emph{outgoing-arc condition}: there is a present
arc $v\to u$ with $x(u)<x(v)$.
\item The \emph{incoming-arc condition}: there is a present
arc $u\to v$ with $y(u)<y(v)$.
\end{itemize}
Using brackets for Boolean truth values, define the outputs
\[
\begin{aligned}
O_v&=\bigvee_{v\to u\in A}
       \bigl(a_{vu}\land[x(u)<x(v)]\bigr),\\
I_v&=\bigvee_{u\to v\in A}
       \bigl(a_{uv}\land[y(u)<y(v)]\bigr),\\
C_{\mathrm{SC}}(a,x,y)
   &=[x(r)=0]\land[y(r)=0]
     \land\bigwedge_{v\neq r}(O_v\land I_v).
\end{aligned}
\]
An empty disjunction is false. Thus a nonroot vertex without
a usable arc in either direction causes the circuit to reject.

We claim that $H_a$ is strongly connected if and only if
there exist labelings $x,y$ for which
$C_{\mathrm{SC}}(a,x,y)=1$. If $H_a$ is strongly connected,
choose $x(v)$ and $y(v)$ as the numbers of arcs on shortest
paths from $v$ to $r$ and from $r$ to $v$, respectively.
These numbers are at most $n-1$ and fit in the allotted bits.
Taking one step forward along a shortest path to $r$ decreases
$x$ by one. Taking one step backward along a shortest path
from $r$ decreases $y$ by one. Both conditions therefore hold.

Conversely, suppose the circuit accepts. Starting at any
nonroot vertex, repeatedly follow an outgoing arc supplied
by its outgoing-arc condition. The nonnegative label $x$
strictly decreases, so this process must terminate. Since
every nonroot vertex has a next step, it can terminate only
at $r$. Likewise, tracing incoming arcs backward strictly
decreases $y$ and reaches $r$, giving a directed path from
$r$ to the starting vertex. Every vertex therefore reaches
$r$ and is reachable from $r$, proving strong connectivity.
For a one-vertex graph, we simply
use the constant output $1$, with no label inputs.

For completeness, we describe the gates implementing the
tests above. For $t\geq1$, a conjunction or disjunction of $t$ inputs can
be formed by a chain of two-input gates using $O(t)$ gates.
Empty conjunctions and disjunctions are the constants $1$
and $0$, respectively. For two $h$-bit integers
$\alpha,\beta$, with bit $0$ the least significant,
define the bit-equality outputs
\[
e_j=(\alpha_j\land\beta_j)
       \lor(\neg\alpha_j\land\neg\beta_j).
\]
Their equality is computed by
$\operatorname{EQ}(\alpha,\beta)=\bigwedge_{j=0}^{h-1}e_j$.
Their strict comparison is computed by the recurrence
\[
\lambda_0=0,\qquad
\lambda_{j+1}=(\neg\alpha_j\land\beta_j)
                 \lor(e_j\land\lambda_j)
\quad(0\leq j<h).
\]
Here $\lambda_j$ compares the lowest $j$ bits: a difference
at the next, more significant bit determines the comparison,
and equality at that bit preserves the previous result.
Thus $\operatorname{LT}(\alpha,\beta)=\lambda_h$ is true
exactly when $\alpha<\beta$. Both tests use $O(h)$ gates.
Together with the AND and OR operations displayed above,
they implement $C_{\mathrm{SC}}$ using
$O((|A|+n)\log n)$ gates. All OR gates are replaced by
constantly many AND and NOT gates as stipulated above.

\proofstep{Encoding a GSCA instance.}
Write $R_1,\ldots,R_q$ for the arc groups of the instance.
By the definition of GSCA, the endpoints of every group lie
in a single weakly connected component of the base graph.
Adding groups therefore does not join different base
components, and we impose strong connectivity separately
in each one. If $q\leq k$, adding all groups and checking
these components determines whether the instance has a
solution. In this case, output a constant-size satisfiable
or unsatisfiable formula accordingly. Assume henceforth
that $q>k$.

Use $k$ blocks of $\ell=\lceil\log_2(q+1)\rceil$ input
bits to record group identifiers $b_1,\ldots,b_k$.
A nonzero identifier $i$ selects $R_i$, and $0$ leaves
that position unused. Repeated identifiers are allowed.
For each position, the range-check output
\[
B_j=\neg\operatorname{LT}(q,b_j)
\]
ensures $b_j\leq q$. Nonnegativity is automatic from the
binary representation. The constant $q$ is represented
using the same $\ell$ bits. Valid identifiers thus encode
exactly all selections of at most $k$ distinct groups.

For each group, compute the selection output
\[
g_i=\bigvee_{j=1}^{k}\operatorname{EQ}(b_j,i).
\]
For each distinct arc $u\to v$ occurring as a base arc
or in a group, compute its presence output
\[
a_{uv}=
\begin{cases}
1, & \text{if $u\to v$ is a base arc},\\
\displaystyle\bigvee_{i:\,(u\to v)\in R_i}g_i,
   & \text{otherwise}.
\end{cases}
\]
The values $g_i$ and $a_{uv}$ are intermediate gate outputs,
determined by the identifier inputs. Each $g_i$ is shared
by all arcs of its group, and each $a_{uv}$ is used wherever
that arc occurs in a connectivity test.

Feed these arc-presence outputs into one copy of
$C_{\mathrm{SC}}$ for each base component, with a fixed
root and its own label inputs. Let $C_Q$ denote this copy
for a component $Q$. The final output is
\[
C_{\mathrm{GSCA}}
   =\bigwedge_{j=1}^{k}B_j
      \land\bigwedge_{Q}C_Q,
\]
where the second conjunction ranges over all base components.
This is an acyclic circuit: it first computes group selection,
then arc presence, and finally the connectivity conditions.
Its free inputs consist only of the identifiers and labels.

If the GSCA instance has a solution, put the identifiers of
its selected groups into the $k$ positions and fill the rest
with zero. The computed arc-presence outputs describe the
augmented graph exactly. Every base component is strongly
connected, so shortest-path labels make each $C_Q$ accept.
The range checks also hold, and $C_{\mathrm{GSCA}}=1$.
Conversely, an accepting assignment has valid identifiers
and hence selects at most $k$ distinct groups. The outputs
$a_{uv}$ describe precisely the graph obtained by adding
these groups. Since every $C_Q$ accepts, the first part of
the proof shows that every base component becomes strongly
connected. The selected groups therefore form a GSCA solution.

\proofstep{Conversion to $3$-CNF and size analysis.}
We count the gates in the circuit before applying
\Cref{lem:circuit-cnf}. There are $kq$ equality tests in
the group-selection outputs, each using $O(\log(q+1))$
gates. The disjunctions defining all $g_i$ use $O(kq)$
further gates. Since $q=O(N^2)$, the total is
\[
O\bigl(kq\log(q+1)\bigr)=O(kN^2\log N).
\]
The $k$ range checks use $O(k\log N)$ gates. There are
$O(N^2)$ group-arc entries in total, so all disjunctions
defining arc presence together use $O(N^2)$ gates, even
when groups share arcs.

There are also $O(N^2)$ distinct base or available arcs.
The connectivity circuits together therefore use
$O(N^2\log N)$ gates: each arc supplies at most two
comparisons of $O(\log N)$-bit labels, and combining the
tests at their endpoints and fixing the root labels uses
$O(N^2+N\log N)$ further gates. Conjoining all component
outputs and range checks uses $O(N+k)$ gates. The entire
circuit consequently has $O(kN^2\log N)$ gates.

By \Cref{lem:circuit-cnf}, this circuit yields a $3$-CNF
formula $\Phi$ with $O(kN^2\log N)$ literal occurrences.
In particular, the encoding requires the circuit output
to be true, so $\Phi$ is satisfiable if and only if the
circuit has an accepting input assignment. By the second
part of the proof, this holds if and only if the GSCA
instance has a solution. Finally, $k<q$ ensures that the
construction takes deterministic polynomial time in the
input size, even when the budget is encoded in binary.
\end{proof}

\subsection{Reduction to Planar \texorpdfstring{$3$}{3}-CNF}
The next step makes the formula suitable for a planar graph
construction. Its \emph{incidence graph} has one vertex for each
variable and each clause, with an edge for every occurrence
of a variable in a clause. A formula is \emph{planar} if this
graph is planar. We need both an embedding and a bound on
the size increase. The following statement follows from
Lichtenstein's crossover construction~\cite{Lichtenstein1982}.

\begin{theorem}\label{thm:planarization}
A $3$-CNF formula with $m$ literal occurrences can be transformed
in polynomial time into an equisatisfiable
$3$-CNF formula with a planar incidence graph and
$O(m^2)$ occurrences, together with a plane embedding.
\end{theorem}
\begin{proof}
Remove tautological clauses, repeated literals within clauses,
and unused variables. An empty clause or the absence of clauses
allows us to output a fixed unsatisfiable or satisfiable planar
formula, respectively. Otherwise, $m\geq1$, and the incidence
graph has $O(m)$ vertices and at most $m$ edges. Draw it in general
position with $c=O(m^2)$ crossings, for example by a
straight-line drawing with a perturbation to separate crossings.

The crossover construction of Lichtenstein has a constant-size
$3$-CNF formula with four distinguished variables
$a,b,a',b'$ appearing in this cyclic order on the outer face
of its incidence graph. An assignment to these variables
extends to a satisfying assignment precisely when
$a=a'$ and $b=b'$. Thus it transmits two independent Boolean
values across a crossing, including both choices for each value.

Place one copy of this construction at each crossing. Along
each remaining segment of an incidence edge, transmit the
value between its endpoint variables $u,v$ using the two clauses
$(\neg u\lor v)$ and $(u\lor\neg v)$, which enforce $u=v$
and can be drawn inside a narrow strip around the segment.
At the clause end of each original incidence edge, use a fresh
copy of its variable and retain the sign of the original literal.
At the variable end, use the original variable. The cyclic order
of the crossover terminals allows all strips to be attached
without crossings, giving an explicit plane embedding.

Every satisfying assignment of the original formula extends
by copying variable values along these connections and then
satisfying the crossover formulas. Conversely, the equality
clauses and crossover formulas force each copied literal to
have the value of its original variable. Restricting any
satisfying assignment therefore satisfies the original formula.
There are $c$ crossover formulas and at most $m+2c$ edge segments,
each contributing constantly many occurrences, in addition
to at most $m$ original clause occurrences. The total is
$O(m+c)=O(m^2)$, and the construction is deterministic
and polynomial in the input size.
\end{proof}

Apply \Cref{thm:planarization} to the formula $\Phi$ constructed
in the preceding subsection, and denote the resulting planar
formula by $\Phi_{\mathrm p}$. If $\Phi$ has $m$ literal
occurrences, then $\Phi_{\mathrm p}$ has $O(m^2+1)$ occurrences, which is $O(m^2)$ for any nontrivial input formula.

Removing tautological clauses, repeated literals within clauses,
and unused variables preserves satisfiability and planarity
without increasing the size. An empty clause or the absence
of clauses allows the answer to be determined directly.
Otherwise, every clause contains one to three literals on
distinct variables.

\subsection{Reduction Back to PDFVS}
We now turn the planar formula into a planar digraph. One
pair of vertices represents each literal occurrence, and
variable triangles force a consistent choice across all
occurrences of the same variable. Clause triangles require
these choices to satisfy every clause. The budget leaves
room for only one deletion per occurrence. This will force
a satisfying assignment.

\begin{lemma}\label{lem:return}
Let $\Psi$ be a $3$-CNF formula with at least one clause, with
each clause containing one to three literals on distinct
variables. Given a plane embedding of its incidence graph,
$\Psi$ can be converted in polynomial time into an equivalent
PDFVS instance with at most $6M$ vertices, at most $9M$ arcs,
and parameter $k'=M$, where $M$ is the number of literal
occurrences in $\Psi$.
\end{lemma}
\begin{proof}
\proofstep{Construction.}
Let $B$ be the incidence graph of the input formula, and fix
its plane embedding. For each variable $x$, let $d_x$ be its
number of occurrences in the formula. Following the order of
edges around the variable vertex $x$ in $B$, create a pair of
vertices $t_i,f_i$ for each occurrence, where $1\leq i\leq d_x$.
These indices are local to $x$: different variables use disjoint
vertex sets. Variables with no occurrences can be discarded.
Deleting $t_i$ represents assigning true to the variable at
this occurrence, and deleting $f_i$ represents assigning false.

Join these vertices into an undirected cycle in the order
\[
t_1,f_1,t_2,f_2,\ldots,t_{d_x},f_{d_x}.
\]
For $d_x\geq2$, interpret the indices cyclically, so that
$t_{d_x+1}=t_1$ and $f_{d_x+1}=f_1$.
If $d_x=1$, create only the edge $t_1f_1$. These edges will
ensure that all occurrences of the same variable receive
consistent values. Replace each undirected edge $e=uv$ by
the directed triangle
\[
u\to v\to p_e\to u,
\]
where $p_e$ is a new vertex added only for this edge and used
in no other part of the construction.

\begin{figure}[ht]
\centering
\begin{tikzpicture}
\node[keep] (t1) at (0,1) {$t_1$};\node[v] (f1) at (1.5,1) {$f_1$};
\node[keep] (t2) at (2.2,-.1) {$t_2$};\node[v] (f2) at (1.5,-1.2) {$f_2$};
\node[keep] (t3) at (0,-1.2) {$t_3$};\node[v] (f3) at (-.7,-.1) {$f_3$};
\draw[thick] (t1)--(f1)--(t2)--(f2)--(t3)--(f3)--(t1);
\node[align=center] at (.75,-1.9) {one pair per occurrence\\choose all $t_i$ or all $f_i$};
\node[v] (u) at (5,1) {$u$};\node[v] (v) at (7,1) {$v$};\node[v] (p) at (6,-.7) {$p_e$};
\draw[arr] (u)--(v);\draw[arr] (v)--(p);\draw[arr] (p)--(u);
\node[align=center] at (6,-1.9) {replace every cycle edge $uv$\\by $u\to v\to p_e\to u$};
\end{tikzpicture}
\caption{The left cycle specifies the pairs that must be covered. The right triangle implements one such edge by a directed cycle with a private vertex. Literal occurrences connect to their designated $t_i$ or $f_i$.}
\label{fig:variable}
\end{figure}
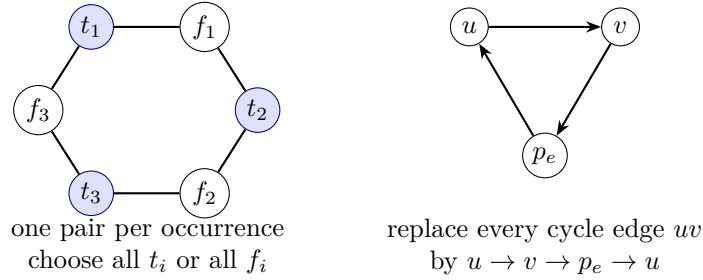

Next, for each clause, choose one vertex from the pair
corresponding to each of its literals. If the literal is the
$i$th occurrence of $x$, choose $t_i$ for the positive literal
$x$ and $f_i$ for the negative literal $\neg x$. Join the chosen
vertices into a directed triangle, following the cyclic order
of its incidence edges. For a clause with one or
two literals, add two or one new vertices, respectively, used
only for that clause, to complete the triangle. Different
occurrences use different vertices, so the clause triangles
are vertex-disjoint.

Figure~\ref{fig:clause-routing} shows the connection for a clause
$C=(x\lor\neg y\lor z)$. Write $\ell_x=t_i$, $\ell_y=f_j$,
and $\ell_z=t_h$ for its designated occurrence vertices in the
three variable gadgets. The clause contributes the three arcs
$\ell_x\to\ell_y\to\ell_z\to\ell_x$.

\begin{figure}[htbp]
\centering
\begin{tikzpicture}[font=\small]
\node at (4.5,4.05) {$C=(x\lor\neg y\lor z)$};
\begin{scope}[shift={(1.5,0)}]
  \node[v] (cx) at (0,1.65) {$x$};
  \node[v] (cy) at (-1.3,-.85) {$y$};
  \node[v] (cz) at (1.3,-.85) {$z$};
  \node[draw,rectangle,fill=gray!10,inner sep=4pt] (cc) at (0,0) {$C$};
  \draw[thick,gray!75] (cc)--(cx);
  \draw[thick,gray!75] (cc)--(cy);
  \draw[thick,gray!75] (cc)--(cz);
  \node[align=center] at (0,3.50) {(a) Incidence graph};
  \node[align=center] at (0,-2.25) {one incidence edge\\per literal occurrence};
\end{scope}
\begin{scope}[shift={(7.1,0)}]
  % Gray corridors follow the original incidence edges.
  \draw[gray!18,line width=16pt,line cap=round] (0,0)--(0,2.13);
  \draw[gray!18,line width=16pt,line cap=round] (0,0)--(-2,-1.2);
  \draw[gray!18,line width=16pt,line cap=round] (0,0)--(2,-1.2);
  \draw[draw=blue!40,fill=blue!3] (0,2.58) circle (.50);
  \draw[draw=blue!40,fill=blue!3] (-2.4,-1.47) circle (.50);
  \draw[draw=blue!40,fill=blue!3] (2.4,-1.47) circle (.50);
  \node at (0,2.72) {$x$};
  \node at (-2.48,-1.63) {$y$};
  \node at (2.48,-1.63) {$z$};
  \draw[gray!65,dashed,fill=white] (0,0) circle (.62);
  \node[font=\scriptsize,align=center,text=gray!80] at (0,.05)
    {clause\\region};
  \node[keep,minimum size=6mm] (lx) at (0,2.13) {$\ell_x$};
  \node[keep,minimum size=6mm] (ly) at (-2,-1.2) {$\ell_y$};
  \node[keep,minimum size=6mm] (lz) at (2,-1.2) {$\ell_z$};
  \draw[arr,blue!70!black] (lx)
    .. controls (-.17,1.65) and (-.15,.84) .. (-.17,.49)
    .. controls (-.19,.24) and (-.40,-.03) .. (-.65,-.20)
    -- (ly);
  \draw[arr,blue!70!black] (ly)
    -- (-.43,-.46)
    .. controls (-.18,-.61) and (.18,-.61) .. (.43,-.46)
    -- (lz);
  \draw[arr,blue!70!black] (lz)
    -- (.65,-.20)
    .. controls (.40,-.03) and (.19,.24) .. (.17,.49)
    .. controls (.15,.84) and (.17,1.65) .. (lx);
  \node[align=center] at (0,3.50) {(b) Routed clause triangle};
  \node[align=center] at (0,-2.25)
    {$\ell_x=t_i,\quad\ell_y=f_j,\quad\ell_z=t_h$};
\end{scope}
\end{tikzpicture}
\caption{Replacing a three-literal clause by a directed triangle.
The blue arcs use opposite sides of the gray incidence corridors
and join without crossing inside the dashed clause region.
Each blue curve is one arc, with no subdivision vertices. The
clause region contains no graph vertex. The variable gadgets
are omitted except for their designated occurrence vertices.
Deleting $\ell_x$, $\ell_y$, or $\ell_z$ represents satisfying
the corresponding literal and breaks this clause cycle.}
\label{fig:clause-routing}
\end{figure}
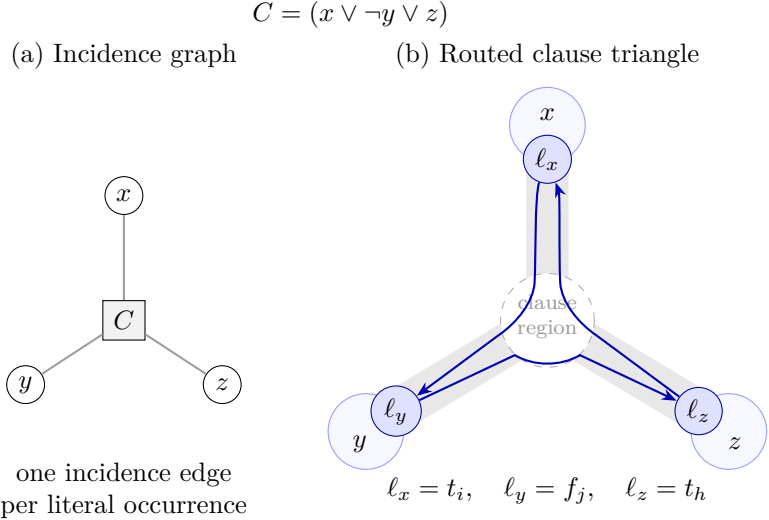

The formula has $M=\sum_x d_x$ literal occurrences, so the
construction has $M$ pairs $t_i,f_i$. Set the output parameter
to $k'=M$.

\proofstep{From a satisfying assignment to a feedback vertex set.}
Delete all $t_i$ for each true variable and all $f_i$ for
each false variable, deleting $M$ vertices in total. Every
edge $e=uv$ of a variable cycle, including the single edge
when $d_x=1$, has a deleted endpoint. Consequently, its direct
arc $u\to v$ disappears. Its private vertex $p_e$ has only
the incoming arc $v\to p_e$ and outgoing arc $p_e\to u$;
at least one of them disappears, so $p_e$ cannot lie on any
remaining directed cycle.

It follows that a remaining directed cycle could use only
clause arcs. Different clause triangles are vertex-disjoint,
so such a cycle would have to lie entirely in one clause
triangle. But every clause has a true literal, and its
designated vertex has been deleted. Thus every clause
triangle is broken, and no directed cycle remains. This also
rules out cycles that might otherwise pass through several
parts of the construction.

\proofstep{From a feedback vertex set to a satisfying assignment.}
Let $X$ be a feedback vertex set of size at most $M$.
First eliminate private vertices of variable triangles from
$X$. For a triangle $u\to v\to p_e\to u$ with $p_e\in X$,
replace $p_e$ by $u$, obtaining
$(X\setminus\{p_e\})\cup\{u\}$. Every directed cycle
containing $p_e$ also contains $u$, because $u$ is the only
out-neighbor of $p_e$. The replacement therefore preserves
the feedback property and does not increase the size. Repeat
until no private variable vertex belongs to $X$.

For each occurrence, the triangle corresponding to the edge
$t_if_i$ must still be hit. Its private vertex is not in $X$,
so $X$ contains at least one of $t_i,f_i$. There are $M$
pairwise disjoint occurrence pairs and $|X|\leq M$. Hence
$|X|=M$, exactly one vertex of each pair belongs to $X$, and
no private clause vertex belongs to $X$.

We next show consistency for each variable. If $d_x\geq2$
and some $t_i$ belongs to $X$, then $f_i$ does not. The
triangle corresponding to the edge $f_it_{i+1}$ must be hit.
Neither $f_i$ nor the private vertex of this triangle belongs to $X$, so
$t_{i+1}\in X$. Iterating cyclically shows that all $t_i$
belong to $X$. If no $t_i$ belongs to $X$, then every $f_i$
does. For $d_x=1$, the same conclusion follows directly from
the one-per-pair condition. Assign $x$ true in the first case
and false in the second. Every clause triangle must be hit
at a designated literal vertex, because none of its private
vertices is in $X$. By the choice of $t_i$ for a positive
literal and $f_i$ for a negative literal, that literal is
true under the assignment. Thus every clause is satisfied.

\proofstep{Planarity.}
We realize the construction in the plane embedding of $B$.
Choose disjoint small disks around its variable and
clause vertices, and disjoint narrow corridors along its edges
outside these disks. Inside each variable disk, place the
occurrence pairs in the cyclic order of the incidence edges.
Draw the variable cycle and its private triangles toward the
interior, leaving each designated literal vertex accessible
from its incidence corridor.

For a three-literal clause, route its directed triangle as in
Figure~\ref{fig:clause-routing}: the two arcs incident with
each literal vertex use opposite sides of the corresponding
corridor. Inside the clause disk, connect consecutive corridors
in their cyclic order. These connections do not cross, and
they introduce no additional vertices. For a one- or two-literal
clause, place its private vertices inside the clause disk and
complete the triangle there in the same way. Distinct clauses
use distinct disks and incidence corridors. Thus all triangles
can be drawn simultaneously without crossings, proving planarity.

\proofstep{Size and running time.}
The variable part starts with $2M$ vertices and at most $2M$
edges to replace. Replacing each edge adds one vertex and
produces three arcs, so this part has at most $4M$ vertices
and $6M$ arcs. There are at most $M$ clauses, each adding at
most two vertices and three arcs. The total is at most $6M$
vertices and $9M$ arcs, and the entire construction takes
polynomial time.
\end{proof}

\subsection{Proof of the Main Theorem}
It remains to combine the three reductions and express their
size bounds in terms of the original parameter. In this
subsection, $k$ denotes the original input budget and $k_1$
the remaining budget after the primal reductions, as in the
overview of the algorithm.

\begin{proof}[Proof of \Cref{thm:main}]
If the primal reductions determine the answer, the algorithm returns YES or NO directly. Otherwise, $2\leq k_1\leq k$.
The primal reductions, dual conversion, and distance-record
merging yield an equivalent GSCA instance with budget $k_1$.
The safety of the primal reductions follows from
Section~\ref{sec:primal}, while the remaining equivalences follow
from \Cref{lem:dual,lem:segments,lem:safety}.

The instance consists of vertices, base arcs, ordinary arc
groups, and special arc groups. Let $W_i$ be the retained set
in the $i$th weakly connected component, put $s_i=|W_i|$,
and let $N_i$ be its number of vertices after merging.
Before merging, each base component is acyclic and has a source,
so $s_i\geq1$.
By \Cref{lem:retained}, $\sum_i s_i=O(k^3)$. Together with
\eqref{eq:quotientorder}, this gives
\begin{equation}\label{eq:totalsize}
\sum_iN_i
=O\left(k^4\sum_i s_i^4\right)
\leq O\left(k^4\left(\sum_i s_i\right)^4\right)
=O(k^{16}).
\end{equation}
Base arcs and ordinary arc groups are counted by ordered
endpoint pairs, so each has total number
$O(\sum_iN_i^2)=O(k^{32})$. By \Cref{lem:retained}, the
special arc groups have $O(k^3)$ arc entries in total.
Thus the instance has $O(k^{16})$ vertices, $O(k^{32})$ arc
groups, and $O(k^{32})$ base arcs and group-arc entries in total.

An instance with fewer than two vertices is already decided.
Otherwise, choose an integer $N\geq2$ that bounds the number
of vertices and whose square bounds both the number of groups
and the total number of base arcs and group-arc entries. The
preceding estimates allow $N=O(k^{16})$. Here $N$ is a common
size bound, rather than necessarily the actual vertex count.
Applying \Cref{lem:cnf} with budget $k_1$ gives a $3$-CNF
formula with
\[
m=O(k_1N^2\log N)=O(k^{33}\log k)
\]
literal occurrences. \Cref{thm:planarization} converts it
into an equisatisfiable planar $3$-CNF formula whose number
of literal occurrences is
\[
M=O(m^2+1)=O(k^{66}\log^2 k).
\]
After the formula preprocessing described above, a directly
decided formula yields a fixed YES- or NO-instance. Otherwise,
\Cref{lem:return} gives an equivalent PDFVS instance
with at most $6M$ vertices, $9M$ arcs, and parameter $k'=M$.

Every step preserves the answer and takes deterministic
polynomial time. The output vertex count, arc count, and
parameter are therefore all $O(k^{66}\log^2 k)$, giving
the claimed polynomial kernel.
\end{proof}

\section{Conclusion}
We obtain a polynomial kernel for planar directed feedback
vertex set by combining structural reductions in the primal
graph with compression of an equivalent dual augmentation
instance. The main technical contribution is the compression
of GSCA: distance records bound the number of vertex
classes, and the cut-tree argument proves that identifying
each class preserves feasibility.

\paragraph{Improving the kernel size.}
The compressed GSCA instance has $O(k^{16})$ vertices and
$O(k^{32})$ base arcs and group-arc entries in total, giving
a polynomial compression of $O(k^{32}\log k)$ bits under an
explicit binary encoding. To obtain a kernel for PDFVS, we convert this GSCA instance back to the original problem.
Since GSCA belongs to $\mathrm{NP}$ and PDFVS is NP-complete,
standard reductions imply the existence of a
polynomial-size return to PDFVS. Our explicit construction
through $3$-CNF, formula planarization, and directed triangles
gives a PDFVS kernel whose vertex count, arc count, and output parameter are all $O(k^{66}\log^2 k)$. Our priority here is to establish the existence of a polynomial kernel through an explicit return construction. A more
size-efficient return reduction could improve the bound
without changing the structural compression.

\paragraph{Implications for single-exponential algorithms.}
We also investigated single-exponential parameterized algorithms
for PDFVS, guided by the same distinction between ordinary and
special vertices as in our kernelization. Since Lokshtanov,
Saurabh, and Xue have already obtained such
algorithms~\cite{LokshtanovSaurabhXue2026}, we restrict our
discussion here to a general face-counting lemma and its
algorithmic implications. We state and prove this lemma in
Appendix~\ref{app:small-face}.

The lemma gives a criterion for finding an unmarked directed
face with few free special vertices. Here a vertex is
\emph{free} if it is eligible for deletion, and a face is
marked when all its special boundary
vertices are protected. If $s$ counts the free special vertices
and $q$ counts the marked directed faces, then, for every
integer $d\geq4$, the condition
\[
s\geq\frac{d-1}{d-3}q
\]
guarantees an unmarked directed face with at most $d$ free
special vertices. The appendix also explains how this face
supplies a branching set in the deterministic framework of
Lokshtanov, Saurabh, and Xue~\cite{LokshtanovSaurabhXue2026}.

Our criterion extends their formally proved $d=4$ case to
arbitrary $d\geq4$, using the same planar incidence-counting
argument~\cite[Lemma~5.1]{LokshtanovSaurabhXue2026}.
Their paper establishes a deterministic $O^*(8.030^k)$-time
algorithm and faster randomized algorithms, and its introduction
also announces improvements using larger branching
sets~\cite[Sections~1, 4, and~5]{LokshtanovSaurabhXue2026}.
We believe the general criterion may help in designing faster
deterministic single-exponential algorithms through an adaptive
choice of $d$ and a refined running-time analysis.

\paragraph{Connections to directed cut problems.}
Another direction is to investigate whether this compression
can help with terminal-separation problems on planar digraphs.
A closely related question was posed by Wlodarczyk at
Dagstuhl Seminar 24411~\cite[Section~4.8]{Dagstuhl24411}.
The \textsc{Skew Multicut} variant considered there asks for
at most $k$ nonterminal vertices whose deletion destroys every
path from an earlier terminal to a later one in a prescribed
order. The question asks for a polynomial kernel on planar
directed acyclic graphs, parameterized by $k$ plus the number
of terminals, and further asks whether the compression can be
performed before the terminal order is supplied. Such a
compression, valid for every subsequent order, was proposed
as a route towards a polynomial kernel for planar DFVS.

Our result settles the planar DFVS question through a different
approach; it does not establish this stronger preservation of
terminal-order constraints or directly yield kernels for
\textsc{Directed Multiway Cut} or \textsc{Directed Multicut}
on planar digraphs. One possible route is to adapt the distance
records and cut-tree argument to preserve terminal-separation
requirements, beyond the augmentation feasibility needed here.
Alternatively, a polynomial-time reduction from a cut problem
to planar DFVS, with deletion budget polynomial in the chosen
cut parameter, would allow our kernel to provide a polynomial
compression. A kernel for the original cut problem would then
require a suitable reduction back. Establishing such reductions
or extending the preservation argument remains a separate task.

More broadly, whether DFVS admits a polynomial kernel
parameterized by the deletion budget $k$ on general digraphs
remains a longstanding open
problem~\cite{FominEtAl2019,DirksEtAl2025}.

\section*{Statement on the Use of AI}

In our earlier work on single-exponential parameterized algorithms and
polynomial kernels for PDFVS, we developed several reduction rules and
structural bounds for planar graphs, including a bound on the number of
special vertices. These results make up the structural simplification
of the primal graph in Section~3.

To further bound the number of ordinary vertices, we tried several
approaches, including studying the problem in the dual graph. At that
stage, we mainly considered reducing the instance by deleting vertices,
and also tried edge contractions, but these approaches had not yielded
the desired compression. We later shared these attempts and our
existing results with GPT-6 Astra for discussion. In subsequent
discussions, AI proposed merging vertices with identical distance
vectors. Through further discussion and verification, we confirmed
that this merging yields the required polynomial compression, although
it need not preserve planarity.

For the final reduction back to PDFVS, we knew that such a reduction
existed and therefore asked AI to help with the explicit construction.
AI's initial proposal involved many intermediate problems and was
relatively complicated. We then simplified this reduction through
interaction with AI. We did not use AI to further explore how this
reduction could be improved to obtain a smaller kernel.

In writing the paper, we first worked through the proofs proposed by AI
and confirmed their correctness through discussion. We then organized
the paper according to our own understanding. Specifically, we first
provided the proof outlines and logical structure in Chinese, then used
AI to produce the English text, which we revised, checked, and polished.
The authors make the final decisions on the content included in the
paper and take full responsibility for its content and correctness.

\providecommand{\etalchar}[1]{$^{#1}$}

\clearpage
\appendix
\section{Unmarked Directed Faces with Few Free Special Vertices}
\label{app:small-face}

We give a general face-counting criterion that can be used
to choose branching sets of special vertices. The statement
applies independently of the kernelization reductions.

Let $D$ be a nonempty loopless plane digraph that is a disjoint
union of strongly connected components, each with at least two
vertices. We consider each component in its induced embedding
and count its faces separately, including its outer face.
For a vertex $v$, let $\sigma_D(v)$ be the number of changes
between incoming and outgoing arcs in the cyclic order around
$v$. As before, $v$ is \emph{ordinary} if $\sigma_D(v)=2$ and
\emph{special} if $\sigma_D(v)>2$. Strong connectivity ensures
that every vertex has both an incoming and an outgoing arc,
so these two cases cover all vertices.

Let $P\subseteq V(D)$ be a set of \emph{protected vertices},
which are forbidden from being deleted. The remaining vertices
are \emph{free}. Define the set of free special vertices by
\[
S=\{v\in V(D)\setminus P:\sigma_D(v)>2\},\qquad s=|S|.
\]
Let $\mathcal M$ be a set of distinct \emph{marked directed
faces}, each of whose special boundary vertices belongs to
$P$, and put $q=|\mathcal M|$. Directed faces outside
$\mathcal M$ are \emph{unmarked}. Faces from different
components are treated as distinct. When counting vertices
on a face, we count each vertex only once, even if its boundary
walk visits that vertex several times.

\begin{lemma}\label{lem:unmarked-face}
Let $D,P,\mathcal M,s$, and $q$ be as above, and let $d\geq4$
be an integer. If
\[
s\geq\frac{d-1}{d-3}q,
\]
then some unmarked directed face contains at most $d$ free
special vertices on its boundary.
\end{lemma}
\begin{proof}
Let $r\geq1$ be the number of components of $D$, and let $F$
be the total number of their directed faces. The Euler
counting identity~\eqref{eq:angles} applies to each component.
Its derivation requires no kernelization reductions. Each
ordinary vertex contributes zero to the vertex sum, each
special vertex contributes at least one, and each term for
a nondirected face is nonnegative. Summing over the components
therefore gives
\[
F\geq 2r+\sum_{v\in V(D)}
\left(\frac{\sigma_D(v)}2-1\right)\geq s+2r.
\]
Let $u=F-q$ be the number of unmarked directed faces. Since
the hypothesis implies $s\geq q$, we have
\[
u\geq s+2r-q\geq2r>0.
\]

Form a bipartite incidence graph $B$ with one part $S$ and
the other part the unmarked directed faces. Join a vertex
to a face precisely when the vertex lies on its boundary,
keeping one edge per vertex--face pair. This graph is simple
and planar. Indeed, for each component of $D$, start with
its radial incidence multigraph, with one face vertex in
each face and an edge for each boundary incidence. Delete
the unwanted vertices and faces and retain only one edge
per remaining pair. Taking the disjoint union of these
graphs gives $B$. This construction also handles repeated
vertices on facial boundary walks.

Suppose, for a contradiction, that every unmarked directed
face contains at least $d+1$ free special vertices. Since
$u>0$, the graph $B$ then has at least $d+2\geq6$ vertices.
Counting its edges at the face vertices and using the edge
bound for simple planar bipartite graphs yields
\[
(d+1)u\leq |E(B)|\leq2(s+u)-4,
\]
and hence $(d-1)u\leq2s-4$. Substituting the lower bound
$u\geq s+2r-q$ gives
\[
(d-1)(s+2r-q)\leq2s-4.
\]
Rearranging, we obtain
\[
(d-3)s\leq(d-1)q-2r(d-1)-4<(d-1)q,
\]
contrary to $s\geq\frac{d-1}{d-3}q$.
\end{proof}

For $d=4$, the condition is $s\geq3q$, recovering the
four-vertex bound of Lokshtanov, Saurabh, and
Xue~\cite[Lemma~5.1]{LokshtanovSaurabhXue2026}.
For $d=5$ and $d=6$, it becomes $s\geq2q$ and
$s\geq5q/3$, respectively. Thus larger values of $d$ allow
branching to continue with fewer free special vertices
relative to the number of marked faces, at the cost of
potentially more branches.

\paragraph{Algorithmic interpretation.}
We briefly explain how the lemma fits into the deterministic
branching framework of Lokshtanov, Saurabh, and
Xue~\cite[Section~5]{LokshtanovSaurabhXue2026}.
Let $b$ be the remaining deletion budget. At each recursive
call, after the usual budget and acyclicity checks, remove
arcs between distinct strongly connected components and
singleton components, and recompute $S$ in the inherited
embeddings. For an unmarked face supplied by the lemma,
list its free special vertices as $v_1,\ldots,v_t$, where
$t\leq d$. In deletion branch $i$, delete $v_i$, protect
$v_1,\ldots,v_{i-1}$, and decrease $b$ by one. In one additional
branch, protect all $v_1,\ldots,v_t$ and mark the face,
leaving $b$ unchanged. These branches partition the solutions
according to the first listed vertex they contain, with the
last branch covering solutions that avoid them all. Deletions
cannot turn an ordinary vertex into a special one, so future
branching never deletes a marked boundary. Its directed
boundary walk also survives component cleanup. Hence marked
faces persist as distinct directed faces.

Every ordinary vertex is incident with at most two directed
faces, so a solution must use at least $\lceil q/2\rceil$
ordinary vertices to meet the marked faces. Reject if $q>2b$.
Enumerate if $s<3q$ or $b=\lceil q/2\rceil$. Otherwise,
branch with $d=4$, which is permitted by $s\geq3q$.
These choices recover their deterministic branching strategy.
In the enumeration step, put
$U=V(D)\setminus(P\cup S)$, the set of currently free ordinary
vertices. For each $Y\subseteq S$ with
$|Y|\leq b-\lceil q/2\rceil$, delete $Y$ and seek a feedback
vertex set of $D-Y$ of size at most $b-|Y|$ using only vertices
of $U$. Thus all other currently special vertices remain
forbidden, even if they become ordinary after deleting $Y$.
This completion is polynomial-time solvable by weighted vertex
splitting, following the approach of Bai, Cao, and
Xiao~\cite{BaiCaoXiao2026}; see also
\cite[Lemma~3.6]{LokshtanovSaurabhXue2026}.
The full correctness and $O^*(8.030^k)$ running-time analysis
are given in~\cite[Section~5]{LokshtanovSaurabhXue2026}.
For general $d$, \Cref{lem:unmarked-face} supplies additional
branching choices. Selecting among them and establishing an
improved running-time bound require a further analysis.

\end{document}